\documentclass[11pt,a4paper]{article}
\usepackage{amssymb, amsmath,accents}
\usepackage{overpic}
\usepackage{float}

\usepackage{hyperref}  
\usepackage[left=2cm, right=2.5cm]{geometry}
\usepackage{xcolor}
\usepackage{pdfsync}

\usepackage{graphicx}
\graphicspath{{figures/}}
\usepackage{epsfig}
\usepackage{amsmath}
\usepackage{amssymb}
\usepackage{amsthm}
\usepackage{mathtools}
\usepackage{epstopdf}
\usepackage{amsfonts}
\usepackage{lineno}
\usepackage{tikz}
\usepackage{subcaption}
\usetikzlibrary{positioning,shapes}
\usetikzlibrary{arrows}

\allowdisplaybreaks

\newtheorem{theorem}{Theorem}
\newtheorem{definition}[theorem]{Definition}

\newtheorem{lemma}[theorem]{Lemma}
\newtheorem{assumption}[theorem]{Assumption}

\newtheorem{remark}[theorem]{Remark}

\newcommand{\diag}{\mathrm{diag}}%
\newcommand{\trace}{\text{trace}}

\usepackage{xspace}	
\begin{document}
	
	
	\title{
		Robust Finite-Time Consensus Subject to Unknown Communication Time Delays Based on Delay-Dependent Criteria
	}
	
	\author{Maryam~Sharifi\footnote{School of Electrical Engineering and Computer Science, KTH Royal Institute of Technology, Stockholm, Sweden, \text{msharifi@kth.se}.}}
	\date{}		
	
	\maketitle		
	\begin{abstract}
		In this paper, robust finite-time consensus of a group of nonlinear multi-agent systems in the presence of communication time delays is considered. In particular, appropriate delay-dependent strategies which are less conservative are suggested. Sufficient conditions for finite-time consensus in the presence of deterministic and stochastic disturbances are presented. The communication delays don't need to be time invariant, uniform, symmetric, or even known. The only required condition is that all delays satisfy a known upper bound. The consensus algorithm is appropriate for agents with partial access to neighbor agents' signals. The Lyapunov-Razumikhin theorem for finite-time convergence is used to prove the results. Simulation results on a group of mobile robot manipulators as the agents of the system are presented. 
	\end{abstract}			
	
\section{Introduction}
Distributed cooperative control of multi-agent systems has been extensively studied in recent decades due to its applicability in the real physical world. The applications include cooperation between robots, distributed estimation in sensor networks, cyber physical systems, and communication network. In these systems, local controllers provide an appropriate group behavior using information from neighbor agents. Among different types of group coordination problems such as formation, flocking, coverage, and rendezvous, consensus is more inclusive and often is applied to other coordination schemes. The objective of consensus control is to achieve an agreement among agents' states by designing local distributed controllers (see e.g., \cite{shao2018leader}, \cite{shi2019containment} on the consensus of multi-agent systems).
	
	In order to improve the performance of multi-agent systems, some significant challenges should be resolved. Some examples are the rate of convergence, the presence of time-delays in communications, and robustness issues. 
	
	In many physical systems, the transient performance of the system is important. It is appropriate that control objectives are achieved in a finite-time interval which is faster and more precise comparing to Lyapunov-sense stability. Consider a cooperative robotic system that fulfills several tasks  from approaching and grasping a given object to some more complicated and precise tasks like minimally invasive robotic surgeries. In these applications, accomplishing the goal in finite-time is indispensable.  
	 A strategy giving finite settling time can be beneficial as it guarantees that the task has been performed faster and exactly as commanded \cite{denasi2015teleoperated}. The finite-time stabilization concept originates from optimal control problems for finite-time operation dynamical systems \cite{bhat2000finite}.
	Some more recent work in this field include \cite{wu2016global}, \cite{golestani2016robust},   \cite{oliva2017distributed}. 
		Finite-time consensus tracking for a class of uncertain mechanical systems under switching topology is investigated in \cite{cai2017adaptive}.

	In the above references, the problem of the presence of time-delays in communication between the agents has not been considered. However, many real practical systems are affected with time delays because of the finite speed  information processing and transmission between agents and limited bandwidth of channels. Time delays often cause undesirable dynamic behaviors such as oscillation, performance degradation and instability in the system \cite{park2013leader}. For instance, in a networked or tele-operated multi-agent system where agents communicate through a network suffering from communication delays, the whole stability of the system is influenced and one of the main tasks  is to reduce the adverse effects.
 In \cite{sakthivel2019finite},   the problem of finite-time consensus of linear multi-agent systems with input delay by employing on-fragile control scheme under switching network topology is considered. Communication time-delays, disturbances and nonlinear dynamics are not employed in their design.
 \newline 
		Analysis of the consensus problems of networked manipulators operating on an under-actuated dynamic platform in the presence of communication delays is presented in \cite{nguyen2017cooperative}.
		Finite-time consensus in the presence of time-delays using an auxiliary approach based on an input to output framework, is studied in \cite{li2018finite}.  In our recent work (\cite{sharifi2020finite}), the finite-time consensus in the presence of time-delays was considered. However, the nature of delayed consensus is apparent  in these works, and consensus between all agents does not happen synchronously, i.e., the consensus error is defined as $x_i(t)-x_j(t-\tau)$, where $\tau$ is the communication delay and $i, j$ are agent indices. 
		
	There are also some few work considering finite-time stabilization of time-delayed linear systems (see \cite{stojanovic2013robust}, \cite{zheng2019finite}) but none of them
	 are applied for control purposes in nonlinear multi-agent systems.
	
	Another important factor in the performance of a multi-agent system is its robustness against uncertainties and disturbances. Generally, these can be categorized in two groups; deterministic and stochastic. Deterministic disturbances are caused by lack of full knowledge from the system dynamics or external signals and deteriorate the system performance. On the other hand, stochastic disturbances are often caused by the presence of noise in the measurements, stochastic couplings in the interconnected systems, etc. Therefore, providing a method which handles these phenomena is indispensable. In \cite{saboori2014h_}, the consensus problem with $H_\infty$ bounds for a homogeneous team of LTI systems in a directed switching topology is studied. The fixed-time consensus problem of first integrator dynamic agents with input delay and uncertain disturbances via an event-triggered control  is studied in \cite{liu2019fixed}. In another work, \cite{kaviarasan2020non} achieves finite-time consensus for a class of stochastic nonlinear multi-agent systems with input delay. However, communication delays and leader-follower scheme are not considered in them. The distributed estimation for a formation of LTI agents under stochastic communication topology and Gaussian process noise is investigated in \cite{subbotin2009design}. 
	In \cite{yang2010finite}, finite-time stochastic synchronization problem for complex networks with noise perturbations is studied.

	In this paper, in order to confront the effects of time-delays and disturbances, two distributed robust control schemes are provided. One of them considers deterministic disturbances and the other takes into account the stochastic ones. There is no necessity for an agent to possess all the neighbor signals and only a part of them is sufficient under some conditions. Furthermore, there are no limitations on time-delays. They can be time-varying, non-uniform, non-symmetric and even unknown but bounded.  
	Establishment of consensus in multi-agent systems is much more difficult in a finite-time sense especially when there are delays in the transmissions. Because dynamics of the cooperative system in the presence of delays is more complicated, more restrictive conditions are needed for finite-time consensus. 
	
	The main contributions of this paper are threefold: 1) Design of new distributed control algorithms to guarantee finite-time robust consensus in the presence of communication time-delays and disturbances. Both deterministic and stochastic disturbances are considered and appropriate strategies are presented. In the case that deterministic disturbances are involved, an ${H_\infty }$ finite-time consensus approach is provided. Besides, in the presence of stochastic disturbances, finite-time stochastic consensus scheme is suggested, which means finite-time convergence of the expected value of consensus error norm to zero. 2) The consensus algorithms are delay-dependent. Some criteria which are functions of the upper-bound of the delays in the system, guarantee the finite-time consensus. Therefore, the conservativeness is reduced in comparison to the delay-independent results. 3) Thanks to a new consensus error definition, the consensus happens synchronously (i.e., $\mathop {\lim e_i(t) =\lim x_i(t)-x_j(t)= 0}\limits_{t \to t_{f} }$, where $t_f$ is a finite consensus time). In addition, the leader-follower algorithm is discussed, too.
	
	The paper is organized as follows: In Section \ref{pre}, some preliminaries including a short description of graph theory and some relevant lemmas, problem statement, system description and its properties are presented.  $H_\infty$ finite-time consensus algorithm for the multi-agent system in the presence of deterministic disturbances is provided in Section \ref{main}. Finite-time stochastic consensus is considered in Section \ref{stoch}. Furthermore, the simulation results and concluding points are available in Sections \ref{sim} and \ref{conc}, respectively.   
	\setcounter{section}{1}
\section{Preliminaries}\label{pre}
	\subsection{Communication Graph\cite{cheng2016event}:}
	A communication graph $G$ is denoted by $G = ( V, E, A)$, where $V = \{ 1,...,N\} $ is a finite nonempty node set, $E \subseteq V \times V$ is the edge set and $A$ is the adjacency matrix. 
	 The ordered pair $(j,i) \in E$ shows that node $i$ obtains information from node $j$. In other words, $j$ is the neighbor of $i$. The neighbor set of node $i$ is defined as ${N_i} = \{ j|(j,i) \in E\} $. We assume that $G$ does not contain any self loops. The adjacency matrix $A = \left[ {{a_{ij}}} \right] \in {\mathbb{R}^{N \times N}}$ of $G$ is defined as $a_{ij}=1$ if $(j,i) \in E$ and $a_{ij}=0$ otherwise.
	
	The in-degree and out-degree of a node are the number of edges that this node is the termination and origination point for them, respectively. If the in-degree and out-degree are equal for all the nodes, the graph is balanced.
	Assume ${d_i} = \sum\limits_{j = 1}^N {{a_{ij}}} $ is the in-degree of node $i \in V$ and 
	$D = \diag\{ {d_1},...,{d_N}\}  \in {\mathbb{R}^{N \times N}}$. Then, $L = D - A$ is the Laplacian matrix of the graph $G$. When the graph $G$ is undirected, i.e., $(i,j) \in E \Leftrightarrow (j,i) \in E$, $L$ is symmetric.
	
	\subsection{Lemmas and Definitions}
	\begin{definition}\cite{moulay2008finite}
	    	Consider the system below:
		\begin{equation}\label{lemm1equ}
		\left\{ \begin{array}{l}
		\dot x(t) = f(x(t),x(t - \tau )),\\
		x(\theta ) = \phi (\theta ),\forall \theta  \in [ - h,0],
		\end{array} \right.
		\end{equation}
		where $x(t) \in \Omega  \subseteq \mathbb{R}^n$ is the state vector, $\Omega$ is a bounded neighborhood of the origin. $f(x(t),x(t - \tau ))$ is a continuous vector field which satisfies $f(0,0)=0$ with unique solution in forward time, $h>0$ is a constant time-delay and $\phi(\theta)$ is a vector valued initial condition function.
		The system is finite-time stable if 
		\item(1) the system (\ref{lemm1equ}) is stable, and
		\item(2) there exists $\delta >0$ such that for any $\phi \in C_{\delta}$, there exists $0\le T(\phi)< \infty$ for which $x(t,\phi)=0$ for all $t\ge T(\phi)$, where ${C_\delta }: = \left\{ {\phi  \in \pounds_h^n:{{\left\| \phi  \right\|}_{\pounds_h^n}} < \delta } \right\}$. Furthermore, ${T_0}(\phi ) = \inf \{ T(\phi ) \ge 0:x(t,\phi ) = 0,\forall t \ge T(\phi )\} $ is a functional called the settling time of the system (\ref{lemm1equ}). 
	\end{definition}
	
	\begin{lemma}\cite{yang2013finite}\label{lem1}	
	 Consider system (\ref{lemm1equ}). If there exist real numbers $\beta>1$, $k>0$, a class-$\mathcal{K}$ function $\sigma $, and a $\mathcal{C}^1$ Lyapunov function $V(x)$ for system (\ref{lemm1equ}), such that
		\begin{equation}\label{lemcond}
		\left\{ \begin{array}{l}
		\sigma (\left\| x \right\|) \le V(x),\\
		\dot V(x) \le  - k{V^{\frac{1}{\beta}} }(x),x \in \Omega 
		\end{array} \right.
		\end{equation}
		
	hold whenever $V(x(t+\theta)) \le V(x(t))$ for $\theta \in [-h,0]$, then  system (\ref{lemm1equ}) is finite-time stable.
		
		Furthermore, if $\Omega=\mathbb{R}^n$ and $\sigma$ is a class-$\mathcal{K}_\infty$ function, then the origin is globally finite-time stable. In addition, the settling time of the system (\ref{lemm1equ}) satisfies the following equation for all $t \ge 0$, with respect to initial condition $\phi$.
		\begin{equation}\label{lemmsettle}
		{T_0}(\phi ) \le \frac{\beta }{{k(\beta  - 1)}}{V^{\frac{{\beta  - 1}}{\beta }}}(\phi )
		\end{equation}
		
	\end{lemma}
	\begin{lemma}\cite{liao2002delay}\label{lem2}
		Consider real matrices $A$, $B$ and symmetric positive-definite matrix $C$ and a scalar $\varepsilon >0$. It holds that
		\begin{align*}
		{A^\top}B + {B^\top}A \le \varepsilon {A^\top}CA + {\varepsilon ^{ - 1}}{B^\top}{C^{ - 1}}B.
		\end{align*}
	\end{lemma}
	\begin{lemma}\cite{qian2001non}\label{lem3}
		For $x_i \in \mathbb{R} $, $i=1,...,n$, $0 < p  \le  1$, the following inequality holds.
		\begin{align*}
		{\left(\sum\limits_{i = 1}^n {\left| {{x_i}} \right|} \right)^p} \le {\sum\limits_{i = 1}^n {\left| {{x_i}} \right|} ^p} \le {n^{1 - p}}{\left(\sum\limits_{i = 1}^n {\left| {{x_i}} \right|}\right )^p}.
		\end{align*}
	\end{lemma}
	\begin{lemma}\label{lem4}
				\cite{sarikaya2013hermite}:
				Provided that function $f:[a,b] \to \mathbb{R}$ is convex, the following called Hermite-Hadamard inequality holds.
				\begin{equation*}
				f\left(\frac{{a + b}}{2}\right) \le \frac{1}{{b - a}}\int_a^b {f(x)dx}  \le \frac{{f(a) + f(b)}}{2}.
				\end{equation*}
				\end{lemma}
	\subsection{Model Description}
	In this paper, the following class of nonlinear systems is considered for $N$ agents of the multi-agent system:
	\begin{equation}\label{sys dynamisc}
	\begin{array}{l}
	\dot x_{i}(t) =f_{i}(x_{i}(t)) +\varphi_i (x_{i}(t))u_{i}(t),\\
	u_{i}(t)=g(x_{i}(t),x_{j}(t-\tau_{i,j}(t))),
	\end{array}
	\end{equation}
	where $i,j=1,...,N$. The indices $i$ and $j$ represent each agent and its neighbors, respectively. Moreover, $x_{i}(t) \in {\mathbb{R}^n}$ is the state vector and $u_{i}(t) \in {\mathbb{R}^n}$ is the vector of control input signals which is a function of current agent's signals and the delayed signals of the neighbors. The communication delay from agent $i$ to $j$ is $\tau_{i,j}(t)$. In addition, $f_{i}(x_{i}(t))\in {\mathbb{R}^{n}}$ and $\varphi_i (x_{i}(t))\in {\mathbb{R}^{n \times n}}$ are nonlinear functions.
	
	The following assumptions hold for the system (\ref{sys dynamisc}):
	\begin{assumption}\label{ass1}
		The agents have the same dynamic structure as (\ref{sys dynamisc}). However, the nonlinear functions (i.e., $f_i(x_{i}(t))$, $\varphi_i(x_{i}(t))$) can be different.
	\end{assumption}
	\begin{assumption}\label{asss}
		The functions $\varphi_i(x_{i}(t))$ are invertible. 
	\end{assumption}
	\begin{assumption}\label{ass3}
		Communication delays can be time-varying, non-uniform and asymmetric, providing that they satisfy
		$ {\tau _{i,j}}(t) \le d $, where $d$ is a constant.
	\end{assumption}
	\begin{remark}\label{remark11}
		Considering Assumption \ref{ass3}, Lemma \ref{lem1} can be used with $d$ as delay of the system.  
	\end{remark}
	
	In this regard, we can present the complete dynamics of the multi-agent system as follows:
	\begin{equation}\label{whole dyn}
	\begin{array}{l}
	\dot X(t) = F(X(t)) +\Phi (X(t))U(t),
	\end{array}
	\end{equation}
	where
	\begin{equation*}
	\begin{array}{l}
	X(t) = {[x_1^\top(t),...,x_N^\top(t)]^\top},\\[5pt]
	F(X(t)) =  {[f_1(x_1(t))^\top,...,f_N(x_N(t))^\top]^\top},\\[5pt]
	U(t) = {[u_1^\top(t),...,u_N^\top(t)]^\top},\\[5pt]
	\Phi(X)  = \diag(\varphi_1({x_1}(t)), \varphi_2({x_2}(t)),..., \varphi_N({x_N}(t))).
	\end{array}
	\end{equation*}

\section{$H_{\infty}$ Consensus Control}\label{main}
	In this section, a ${H_\infty }$ algorithm based on a consensus error is suggested. One great advantage with this technique is that it allows the designer to tackle the most general form of control architecture wherein explicit accounting of uncertainties, disturbances, actuator constraints and performance measures can be accomplished. Furthermore, this method has an optimization property in itself.
	It is well known that delay-independent criteria often give more conservative results especially in the cases of small delays \cite{gu2003stability}. Therefore, in this paper, delay-dependent finite-time control algorithms with the sufficient conditions will be presented.
	
In this paper, the communication between agents could  be through a connected undirected graph or a balanced strongly connected digraph. 
First, consider the consensus error vector for the multi-agent system as
\begin{equation}\label{consensus vec}
e(t) = (M \otimes {I_n})X(t),
\end{equation}	
where $\otimes$ denotes the  Kronecker product and matrix $M$ is consisting of the left eigenvectors of the Laplacian matrix of multi-agent system corresponding to $N-1$ nonzero eigenvalues. In other words, if the eigenvalues are arranged from smallest value i.e., $\lambda_0=0$ to $\lambda_{N-1}$ with corresponding eigenvectors $v_1$ to $v_{N-1}$, matrix $M$ would be defined as 
\begin{align*}
M = {\left[ {\begin{array}{*{20}{c}}
		{{v_1^\top}}\\
		\vdots \\
		{{v_{N - 1}^\top}}
		\end{array}} \right]_{(N - 1) \times N}}.
\end{align*}   
	 Note that the eigenvectors of the Laplacian matrix of an undirected graph which is symmetric, are perpendicular to each other. It is also known that the eigenvector corresponding to eigenvalue $\lambda_0=0$ is ${\textbf 1_{N \times 1}} = [\begin{array}{*{20}{c}}
	1& \cdots &1
	\end{array}]_{1 \times N}^\top$. Thus, it is concluded that each row of the matrix $M$  is perpendicular to ${\textbf 1_{N \times 1}}$.  Therefore, $e(t)=0$, (i.e., $ (M \otimes {I_n})X(t)=0$) if and only if
	\begin{align*}
	\left[ {\begin{array}{*{20}{c}}
		{{X_1}}\\
		\vdots \\
		{{X_N}}
		\end{array}} \right] = k\left[ {\begin{array}{*{20}{c}}
		1\\
		\vdots \\
		1
		\end{array}} \right],
	\end{align*}
	where $k$ is a real number. Thus, consensus condition will be achieved.	

For the directed graph topology, consider $l = {[{l_1},...,{l_N}]^\top}$ as the left eigenvector of a Laplacian matrix of a strongly connected directed graph associated with the zero eigenvalue. 
		It is easy to show that $0$ is a simple eigenvalue of the matrix $({I_N} - \textbf 1{l^\top})$ with $\textbf 1$ as a right eigenvector and $1$ is the other eigenvalue with multiplicity $N-1$. Provided that the digraph is balanced, $({I_N} - \textbf 1{l^\top})$ is symmetric. Therefore, in this case,
		matrix $M$ is constructed from the matrix $({I_N} - \textbf 1{l^\top	})$ the same manner as it was constructed from the Laplacian matrix of the undirected graph.
	
	\subsection{$H_{\infty}$ Delay Dependent Consensus Algorithm}
In this section, we provide a consensus algorithm in which the closed loop system is robust against the uncertainties. These uncertainties are either due to model mismatches in the agents' dynamics or the presence of external disturbance signals. We consider a generalized model for the system with disturbances.
\begin{equation}\label{disturbed_sys}
\begin{split}
\dot X(t) = F(X(t)) + \Phi (X(t))U(t) + G(X)w(t),
\end{split}
\end{equation}
where $w(t) \in {\mathbb {R}^q}$, $q < nN$, is the  disturbance signal and $G(X)$ is an appropriately dimensioned state-dependent weighting matrix.
\begin{remark}\label{remark04}
It should be noted that uncertain nonlinear dynamic systems can be treated using our approach, too. The uncertain parts of the system dynamics can be modeled by a similar fashion to (\ref{disturbed_sys}). Consider the uncertain structure of the dynamic system components as below:
	\begin{equation}\label{dist_def}
	\begin{array}{l}
	F(X(t)) = {F_{nom}}(X(t)) + \Delta F(X(t)),\\[5pt]
	\Phi (X(t)) = {\Phi _{nom}}(X(t)) + \Delta \Phi (X(t)),
	\end{array}
	\end{equation}
	in which, ${F_{nom}}(X(t))$ is the nominal known part of $F(X(t))$ and $\Delta F(X(t))$ is the unknown bounded part. A similar definition holds for $\Phi (X(t))$. 
	Therefore, $ G(X)w(t) = \Delta F(X(t)) + \Delta \Phi (X(t))U(t)$
	 models the uncertain part of the system as a state-dependent disturbance and the system dynamics is written as
\begin{align}
	\dot X(t) = {F_{nom}}(X(t)) + {\Phi _{nom}} (X(t))U(t) + G(X)w(t).
	\end{align}
	
\end{remark}
In this section, the objective is to design a distributed control protocol for the agents of the system to achieve  the consensus while preserving a desirable disturbance rejection performance.
Consider the following consensus control algorithm for the multi-agent system of (\ref{disturbed_sys}):
	\begin{equation}\label{controller}
		\begin{split}
		&U(t) =- {\Phi ^{ - 1}}(X(t))[F(X(t))\\&+ ({I_N} \otimes {K_1}){\left({X^\top}(t)X(t)\right)^{\frac{{1 - \alpha }}{{2\alpha  - 1}}}}X(t)\\
		&+ ({I_N} \otimes {K_2})(L \otimes {I_n}){\left({X^\top}(t - d)X(t - d )\right)^{\frac{{1 - \alpha }}{{2\alpha  - 1}}}}X(t - d )],
		\end{split}
		\end{equation}
		in which, $L\in {\mathbb{R}^{N \times N}}$ is the Laplacian matrix of the system, $K_1, K_2\in {{\mathbb{R}^{n \times n}}}$ are appropriate feedback matrices which should satisfy some conditions given later. In particular, each agent has a feedback based on its own states via $K_1$ and its neighbors' delayed states via $K_2$. In addition, $\alpha >1$ is a real number. According to Remark \ref{remark04}, in the control signal \eqref{controller}, the terms $F(X(t))$, $\Phi(X(t))$ are substituted by $ {F_{nom}}(X(t))$ and ${\Phi _{nom}} (X(t))$ in the case of uncertain nonlinear dynamic system.

To establish the $H_{\infty}$ consensus for the multi-agent system, consider the following penalty signal $z(t)$ as the performance variable of the system.
\begin{equation}\label{z}
z(t) = e^\top(t){\left(e^\top(t)e(t)\right)^{\frac{{1 - \alpha }}{{2\alpha  - 1}}}}.
\end{equation}
\begin{definition}\label{def1}
	Given positive scalar $\gamma$ as the disturbance attenuation level, the consensus control law of (\ref{controller}) provides global finite-time consensus with a guaranteed ${H_\infty }$ performance $\gamma$ for the multi-agent system of (\ref{disturbed_sys}), if the following requirements hold:
	
	1. The closed-loop multi-agent network with $w=0$ reaches finite-time consensus (i.e., $\mathop {\lim e (t) = 0}\limits_{t \to T_{c} }$) for any initial condition.
	
	2. For any non-zero $w(t)$, the zero state response of the closed-loop system satisfies
	\begin{equation}\label{robust_cond}
	\int\limits_0^T {{{\left\| {z(t)} \right\|}^2}dt}  \le {\gamma ^2}\int\limits_0^T  {{{\left\| {w(t)} \right\|}^2}dt} ,
	\end{equation}
	for $0\le T < \infty  $.
\end{definition}
Consider the following definitions:
\begin{equation}\label{ineq}
\begin{array}{l}
		{R =  - (M \otimes {I_n})({I_N} \otimes {K_1}){{(M \otimes {I_n})}^ + },}\\
		S =  - (M \otimes {I_n})({I_N} \otimes {K_2})(L \otimes {I_n}){(M \otimes {I_n})^ + },\\
		P =  {{{(M \otimes {I_n})}^ + }^\top{{(M \otimes {I_n})}^ + } },
		\end{array}
\end{equation}
 in which, $(.)^+$ is the pseudo inverse of a non square matrix (i.e., ${A^ + } = {{A^\top}(AA^\top)^{ - 1}}$). 
 Since the matrix $M$ consists of the eigenvectors corresponding to the non zero eigenvalues of the Laplacian matrix of the connected graph, matrix $P$ will be positive-definite. (Similar explanations hold for the balanced strongly-connected directed graph topology case).
	\begin{theorem}\label{Theorem1}
Consider the multi-agent system defined in (\ref{disturbed_sys}) with the consensus error definition of (\ref{consensus vec}) satisfying the Assumptions \ref{ass1}-\ref{ass3} and disturbance attenuation $\gamma>1$. Utilizing the control algorithm of (\ref{controller}), $H_\infty$ finite-time consensus (Definition \ref{def1}) in the presence of communication time-delays is achieved if there exist a matrix $Q\ge 0$ and positive scalars $a, b$, such that that the following inequalities hold.
		\begin{equation}\label{ineqq1}
\begin{array}{l}
   (M \otimes {I_n})G(X)G{(X)^\top}{(M \otimes {I_n})^\top} \le Q,
   \end{array}
\end{equation}
	\begin{equation}\label{ineqq2}
\begin{array}{l}
   q = {\lambda _{\max }}\left[ {\frac{\alpha }{{2\alpha  - 1}}{\lambda _{\min }}{{(P)}^{\frac{{1 - \alpha }}{{2\alpha  - 1}}}}(R + {R^\top}) + {a^{ - 1}}{S^\top}S} \right]\\+(\gamma^{2}-1)^{-1}{\lambda _{\max }}(Q) 
  + \frac{{bd}}{2}\left( {{{[n(N - 1)]}^{\frac{{\alpha  - 1}}{\alpha }}} + 1} \right)\\+ a{\left( {\frac{\alpha }{{2\alpha  - 1}}{\lambda _{\max }}{{(P)}^{\frac{{1 - \alpha }}{{2\alpha  - 1}}}}} \right)^2}{\left[ {n(N - 1)} \right]^{\frac{{\alpha  - 1}}{\alpha }}} + 1 < 0,
\end{array}
\end{equation}
	\begin{equation}\label{ineqq3}
		{{\lambda _{\min }}(P)(R + {R^\top}){ - b{I_{n(N - 1)}}}}+ Q \le 0.
		\end{equation}
		Furthermore, the finite consensus time $T_c$, is bounded by
		\begin{equation}\label{cons time}
		{{{T}}_c} \le \frac{{ \alpha }}{{\left|q\right|(\alpha  - 1)}}V{(0)^{\frac{{\alpha  - 1}}{\alpha }}},
		\end{equation}
	where $\left|\cdot\right|$ defines the absolute value.
	\end{theorem}
\begin{proof}
		Consider the following Lyapunov function for the system using the mentioned consensus error.
		\begin{equation}\label{lyap}
		V(t) = {\left({e^\top}(t)e(t)\right)^{\frac{\alpha }{{2\alpha  - 1}}}}.
		\end{equation}
		By taking the derivative of (\ref{lyap}) one can reach
		\begin{equation}\label{lyapder}
\begin{split}
\dot V(t) =& {{\frac{\alpha }{{2\alpha  - 1}}}\left( {{e^\top}(t)e(t)} \right)^{\frac{{1 - \alpha }}{{2\alpha  - 1}}}}\left( {{e^\top}(t)\dot e(t) + {{\dot e}^\top}(t)e(t)} \right),
\end{split}
		\end{equation}
		in which
			\begin{equation}\label{err dyn}
	\begin{split}
\dot e(t) =&(M \otimes {I_n})\dot X(t) = (M \otimes {I_n})[ (F(X(t))\\&+ \Phi (X(t))U(t)) + G(X)w(t)].
	\end{split}
	\end{equation}
		Substitution of (\ref{controller}) in (\ref{err dyn}) gives
		\begin{equation}\label{subs3}
		\begin{split}
		&\dot e(t) = - (M \otimes {I_n})({I_N} \otimes {K_1}){\left({X^\top}(t)X(t) \right)^{\frac{{1 - \alpha }}{{2\alpha  - 1}}}}X(t)\\&
		- (M \otimes {I_n})({I_N} \otimes {K_2})(L \otimes {I_n}){\left({X^\top}(t - d)X(t - d)\right)^{\frac{{1 - \alpha }}{{2\alpha  - 1}}}}\\&\times X(t - d)+ (M \otimes {I_n})G(X)w(t).
		\end{split}	
		\end{equation}
	By replacing $X(t)$ and $X(t-d)$ with $e(t)$ and $e(t-d)$ using (\ref{consensus vec}), we get
		\begin{equation}\label{final ec}
\begin{split}
&\dot e(t)=  - {(M \otimes {I_n})({I_N} \otimes {K_1}){{(M \otimes {I_n})}^ + }} e(t)\\&\times
{\left( {{e^T}(t){{(M \otimes {I_n})}^ + }^\top{{(M \otimes {I_n})}^ + }e(t)} \right)^{\frac{{1 - \alpha }}{{2\alpha  - 1}}}}\\&
 - \left[ {(M \otimes {I_n})({I_N} \otimes {K_2})(L \otimes {I_n})}{{{(M \otimes {I_n})}^ + }} \right]e(t - d)\\&\times({e^\top}(t - d) {(M \otimes {I_n})^ + }^\top{(M \otimes {I_n})^ + }\\&
 \times e(t - d){)^{\frac{{1 - \alpha }}{{2\alpha  - 1}}}}+ (M \otimes {I_n})G(X)w(t).
\end{split}
		\end{equation}
		Substituting (\ref{final ec}) in (\ref{lyapder}) and utilizing the definitions in (\ref{ineq}), one obtains
		\begin{equation}\label{pred err dynN3}
		\begin{split}
&\dot V = \frac{\alpha }{{2\alpha  - 1}}[{e^\top}(t){\left( {{e^\top}(t)e(t)} \right)^{\frac{{1 - \alpha }}{{2\alpha  - 1}}}}Re(t){\left( {{e^\top}(t)Pe(t)} \right)^{\frac{{1 - \alpha }}{{2\alpha  - 1}}}}\\&+ {e^\top}(t){\left( {{e^\top}(t)Pe(t)} \right)^{\frac{{1 - \alpha }}{{2\alpha  - 1}}}}{R^\top}e(t){\left( {{e^\top}(t)e(t)} \right)^{\frac{{1 - \alpha }}{{2\alpha  - 1}}}}\\&+ {e^\top}(t){\left( {{e^\top}(t)e(t)} \right)^{\frac{{1 - \alpha }}{{2\alpha  - 1}}}}Se(t - d)({e^\top}(t - d)Pe(t - d))^{\frac{{1 - \alpha }}{{2\alpha  - 1}}}\\&+ {e^\top}(t - d)
 {\left( {{e^\top}(t - d)Pe(t - d)} \right)^{\frac{{1 - \alpha }}{{2\alpha  - 1}}}}{S^\top}e(t)\\&\times{\left( {{e^\top}(t)e(t)} \right)^{\frac{{1 - \alpha }}{{2\alpha  - 1}}}}
 + {\left( {{e^\top}(t)e(t)} \right)^{\frac{{1 - \alpha }}{{2\alpha  - 1}}}}({e^\top}(t)\\&\times(M \otimes {I_n})G(X)w(t)
 + {w(t)^\top}G{(X)^\top}{(M \otimes {I_n})^\top}e(t))].
		\end{split}
		\end{equation}
		By using Lemma \ref{lem2} for third and fourth terms of (\ref{pred err dynN3}), the following inequality is acquired.
		\begin{equation}\label{pred err dyn3}
		\begin{split}
		&\dot V(t) \le \frac{\alpha }{{2\alpha  - 1}}{\lambda _{\min }}{(P)^{\frac{{1 - \alpha }}{{2\alpha  - 1}}}}{e^\top}(t){({e^\top}(t)e(t))^{\frac{{1 - \alpha }}{{2\alpha  - 1}}}}\\&\times(R + {R^\top})e(t){({e^\top}(t)e(t))^{\frac{{1 - \alpha }}{{2\alpha  - 1}}}}+ \frac{{2\alpha }}{{2\alpha  - 1}}
		{e^\top}(t){({e^\top}(t)e(t))^{\frac{{1 - \alpha }}{{2\alpha  - 1}}}}S\\&\times e(t - d )({e^\top}(t - d ){(M \otimes {I_n})^ + }^\top{(M \otimes {I_n})^ + } e(t - d ))^{\frac{{1 - \alpha }}{{2\alpha  - 1}}}\\&
		+ {\left( {{e^\top}(t)e(t)} \right)^{\frac{{1 - \alpha }}{{2\alpha  - 1}}}}({e^\top}(t)(M \otimes {I_n})G(X)w(t)\\&
 + {w(t)^\top}G{(X)^\top}{(M \otimes {I_n})^\top}e(t))\\&
		\le {e^\top}(t){({e^\top}(t)e(t))^{\frac{{1 - \alpha }}{{2\alpha  - 1}}}}\left[\frac{\alpha }{{2\alpha  - 1}}{\lambda _{\min }}{(P)^{\frac{{1 - \alpha }}{{2\alpha  - 1}}}} (R + {R^\top}) + a^{-1}{S^\top}S\right]\\&\times e(t){({e^\top}(t)e(t))^{\frac{{1 - \alpha }}{{2\alpha  - 1}}}}
		+ a{(\frac{\alpha }{{2\alpha  - 1}}{\lambda _{\max }}{(P)^{\frac{{1 - \alpha }}{{2\alpha  - 1}}}})^2}{e^\top}(t - d )\\&\times{({e^\top}(t - d )e(t - d ))^{\frac{{1 - \alpha }}{{2\alpha  - 1}}}}e(t - d ){({e^\top}(t - d )e(t - d ))^{\frac{{1 - \alpha }}{{2\alpha  - 1}}}}
\\&+ {\left( {{e^\top}(t)e(t)} \right)^{\frac{{1 - \alpha }}{{2\alpha  - 1}}}}({e^\top}(t)(M \otimes {I_n})G(X)w(t)
\\&+ {w(t)^\top}G{(X)^\top}{(M \otimes {I_n})^\top}e(t)).
		\end{split}
		\end{equation}
		Finally, $\dot V(t)$ can be written as
		\begin{equation}\label{pred err dyn4}
		\begin{split}
		&\dot V(t) \le {\lambda _{\max }}[\frac{\alpha }{{2\alpha  - 1}}{\lambda _{\min }}{(P)^{\frac{{1 - \alpha }}{{2\alpha  - 1}}}}(R + {R^\top})+ a^{-1}({S^\top}S)]\\&\times{\left({e^\top}(t)e(t)\right)^{\frac{1}{{2\alpha  - 1}}}}
		+ a{\left(\frac{\alpha }{{2\alpha  - 1}}{\lambda _{\max }}{(P)^{\frac{{1 - \alpha }}{{2\alpha  - 1}}}}\right)^2}\\&\times{\left({e^\top}(t - d )e(t - d )\right)^{\frac{1}{{2\alpha  - 1}}}}	+ {\left( {{e^\top}(t)e(t)} \right)^{\frac{{1 - \alpha }}{{2\alpha  - 1}}}}({e^\top}(t)\\&\times(M \otimes {I_n})G(X)w(t) + {w(t)^\top}G{(X)^\top}{(M \otimes {I_n})^\top}e(t)).
		\end{split}
		\end{equation}
		Furthermore, (\ref{final ec}) can be written as below by moving all terms to one side.
			\begin{equation}\label{e}
			\begin{split}
			&\dot e(t) - Re(t){\left({e^\top}(t)Pe(t)\right)^{\frac{{1 - \alpha }}{{2\alpha  - 1}}}}- Se(t - d)\\&\times{\left({e^\top}(t - d)Pe(t - d)\right)^{\frac{{1 - \alpha }}{{2\alpha  - 1}}}}-(M \otimes {I_n})G(X)w(t) = 0.
			\end{split}
			\end{equation}	
			Multiplying (\ref{e}) by $2\int_{t - d}^t {{e^\top}(s){{\left({e^\top}(s)e(s)\right)}^{\frac{{1 - \alpha }}{{2\alpha  - 1}}}}}ds$ and adding it to the right side of (\ref{pred err dyn4}) and after some manipulations, $\dot V$ can be expressed as
				\begin{equation}\label{4dotv}
			\begin{split}
			&\dot V(t) \le {\lambda _{\max }}\left[\frac{\alpha }{{2\alpha  - 1}}{\lambda _{\min }}{(P)^{\frac{{1 - \alpha }}{{2\alpha  - 1}}}}(R + {R^\top})+ {a^{ - 1}}{S^\top}S \right]\\&\times{({e^\top}(t)e(t))^{\frac{1}{{2\alpha  - 1}}}}+ a{\left(\frac{\alpha }{{2\alpha  - 1}}{\lambda _{\max }}{(P)^{\frac{{1 - \alpha }}{{2\alpha  - 1}}}}\right)^2} \\&\times{\left({e^\top}(t - d)e(t - d)\right)^{\frac{1}{{2\alpha  - 1}}}}+ {\left( {{e^\top}(t)e(t)} \right)^{\frac{{1 - \alpha }}{{2\alpha  - 1}}}}({e^\top}(t)\\&\times(M \otimes {I_n})G(X)w(t) + {w(t)^\top}G{(X)^\top}{(M \otimes {I_n})^\top}e(t)) \\&+ 2\int_{t - d}^t {{e^\top}(s){{\left({e^\top}(s)e(s)\right)}^{\frac{{1 - \alpha }}{{2\alpha  - 1}}}}}\left[ Re(s) {({e^\top}(s)Pe(s))^{\frac{{1 - \alpha }}{{2\alpha  - 1}}}} \right.\\&\left.+ Se(s - d)({e^\top}(s - d)Pe(s - d))^{\frac{{1 - \alpha }}{{2\alpha  - 1}}}\right]ds\\&+ 2\int_{t - d}^t {e^\top}(s) {{\left({e^\top}(s)e(s)\right)}^{\frac{{1 - \alpha }}{{2\alpha  - 1}}}} Q{{{\left({e^\top}(s)e(s)\right)}^{\frac{{1 - \alpha }}{{2\alpha  - 1}}}}{e}(s)}ds \\&+w(t)^\top w(t)
			-\frac{2\alpha -1}{{\alpha}}\left(V(t)-V(t - d)\right).
			\end{split}
			\end{equation}
Furthermore, due to Lemma \ref{lem4}, we get
				\begin{equation}\label{6dotv}
				\begin{split}
				&b\int_{t - d}^t {{e^\top}(s){{\left({e^\top}(s)e(s)\right)}^{\frac{{1 - \alpha }}{{2\alpha  - 1}}}}}{\left({e^\top}(s)e(s)\right)^{\frac{{1 - \alpha }}{{2\alpha  - 1}}}}e(s)\\&
			\le \frac{{bd}}{2}\left[ {{{\left( {{e^\top}(t)e(t)} \right)}^{\frac{1}{{2\alpha  - 1}}}} + {{\left( {{e^\top}(t - d)e(t - d)} \right)}^{\frac{1}{{2\alpha  - 1}}}}} \right],
				\end{split}
				\end{equation}
				where $b$ is a positive real number.
				
					In addition, utilizing Lemma \ref{lem3} and meeting the required condition of Lemma \ref{lem1}, one obtains	
		\begin{equation}\label{inequality}
		\begin{split}
		&{({e^\top}(t - d )e(t - d ))^{\frac{1}{{2\alpha  - 1}}}}= \sum\limits_{i = 1}^{n(N - 1)} {{{[{{({e_i}{{(t - d )}^2})}^{\frac{\alpha }{{2\alpha  - 1}}}}]}^{\frac{1}{\alpha }}}}  \\&\le \left[n(N - 1)\right]^{\frac{{\alpha  - 1}}{\alpha }}{\left[\sum\limits_{i = 1}^{n(N - 1)} {{{({e_i}{{(t - d )}^2})}^{\frac{\alpha }{{2\alpha  - 1}}}}} \right]^{\frac{1}{\alpha }}}\\&
		= \left[n(N - 1)\right]^{\frac{{\alpha  - 1}}{\alpha }}{\left[{({e^T}(t - d )e(t - d ))^{\frac{\alpha }{{2\alpha  - 1}}}}\right]^{\frac{1}{\alpha }}}\\& = \left[n(N - 1)\right]^{\frac{{\alpha  - 1}}{\alpha }}V{(t - d )^{\frac{1}{\alpha }}}\le \left[n(N - 1)\right]^{\frac{{\alpha  - 1}}{\alpha }}V{(t)^{\frac{1}{\alpha }}}.
		\end{split}
		\end{equation}
				Hence,
				\begin{equation}\label{3ddotv}
				\begin{split}
				&- b\int_{t - d}^t {{e^\top}(s){{\left({e^\top}(s)e(s)\right)}^{\frac{{1 - \alpha }}{{2\alpha  - 1}}}}} {({e^\top}(s)e(s))^{\frac{{1 - \alpha }}{{2\alpha  - 1}}}}e(s)\\&+ \frac{{bd}}{2}\left({[n(N - 1)]^{\frac{{\alpha  - 1}}{\alpha }}}+1\right)V{(t)^{^{\frac{1}{\alpha }}}} \ge 0.
				\end{split}
				\end{equation}	
			
			By adding the left side of (\ref{3ddotv}) to the right side of (\ref{4dotv}), using and (\ref{inequality}) and rewriting the integral terms in quadratic form, we get
			\begin{equation}\label{final}
\begin{split}
&\dot V(t)\le  \{ {\lambda _{\max }}\left[ {\frac{\alpha {\lambda _{\min }}{{(P)}^{\frac{{1 - \alpha }}{{2\alpha  - 1}}}}}{{2\alpha  - 1}}(R + {R^\top}) + {a^{ - 1}}{S^\top}S} \right]
 \\& + \frac{{bd}}{2}\left( {{{[n(N - 1)]}^{\frac{{\alpha  - 1}}{\alpha }}} + 1} \right)
 + a{\left( {\frac{\alpha }{{2\alpha  - 1}}{\lambda _{\max }}{{(P)}^{\frac{{1 - \alpha }}{{2\alpha  - 1}}}}} \right)^2}\\&\times{\left[ {n(N - 1)} \right]^{\frac{{\alpha  - 1}}{\alpha }}}\} V{(t)^{\frac{1}{\alpha }}}
+ w(t)^\top w(t)+ {\left( {{e^\top}(t)e(t)} \right)^{\frac{{1 - \alpha }}{{2\alpha  - 1}}}}\\&\times ({e^\top}(t)(M \otimes {I_n})G(X)w(t) + {w(t)^\top}G{(X)^\top}{(M \otimes {I_n})^\top}\\&\times e(t))+\int_{t - d}^t {\zeta {{(s)}^\top}{\rm{E}}\zeta (s)ds}, 
\end{split}
			\end{equation}
			where
			\begin{equation}
			\begin{array}{l}
			\zeta (s) = {\left[ {\begin{array}{*{20}{c}}
					{\Xi (s)}&{\Xi (s - d)}
					\end{array}} \right]^\top},\\
			E = \left[ {\begin{array}{*{20}{c}}
				E_1&{{\lambda _{\max }}(P)S}\\
				{{\lambda _{\max }}(P){S^\top}}&0
				\end{array}} \right],
			\end{array}
			\end{equation}
			in which $E_1={{\lambda _{\min }}(P)(R + {R^\top}){ - b{I_{n(N - 1)}}}}+ Q$ and	$\Xi (s) = {e^\top}(s){({e^\top}(s)e(s))^{\frac{{1 - \alpha }}{{2\alpha  - 1}}}}$.
		By considering the Schur complement Lemma, we have
		\begin{align*}
		    E \le 0 \leftrightarrow {{\lambda _{\min }}(P)(R + {R^\top}){ - b{I_{n(N - 1)}}}}+ Q \le 0.
		\end{align*}

We will prove the following inequality to establish the required robustness for the multi-agent system.
	\begin{equation}\label{rob_cond}
	\begin{split}
	&J=\int_0^T  {[z{{(t)}^\top}z(t) - {\gamma ^2}w{{(t)}^\top}w(t) + \dot V(t)]dt}\\& + V(0) - V(T) \le 0.
	\end{split}
	\end{equation} 
	
	Substituting (\ref{final}) in (\ref{rob_cond}) and assuming all inequalities of  (\ref{ineqq1})-(\ref{ineqq3}) are satisfied, results in
	\begin{equation}
	\begin{split}
J &\le \int_0^T [ {{{\left( {{e^\top}(t)e(t)} \right)}^{\frac{1}{{2\alpha  - 1}}}}}  -({\gamma ^2}-1)w{(t)^\top}w(t)
 \\&+ \{ {\lambda _{\max }}\left[ {\frac{{\alpha {\lambda _{\min }}{{(P)}^{\frac{{1 - \alpha }}{{2\alpha  - 1}}}}}}{{2\alpha  - 1}}(R + {R^\top}) + {a^{ - 1}}{S^\top}S} \right]\\&
 + \frac{{bd}}{2}\left( {{{[n(N - 1)]}^{\frac{{\alpha  - 1}}{\alpha }}} + 1} \right)+ a{\left( {\frac{\alpha {\lambda _{\max }}{{(P)}^{\frac{{1 - \alpha }}{{2\alpha  - 1}}}}}{{2\alpha  - 1}}} \right)^2}\\&\times{\left[ {n(N - 1)} \right]^{\frac{{\alpha  - 1}}{\alpha }}}\} V{(t)^{\frac{1}{\alpha }}}
 + {\left( {{e^\top}(t)e(t)} \right)^{\frac{{1 - \alpha }}{{2\alpha  - 1}}}}({e^\top}(t)\\&\times(M \otimes {I_n})G(X)w(t) + {w(t)^\top}G{(X)^\top}{(M \otimes {I_n})^\top}\\&\times e(t))]dt
 + V(0) - V(T).
 \end{split}
	\end{equation}
	The above inequality can be written as
	\begin{equation}
	\begin{split}
	J &\le \int_0^T [\{ {\lambda _{\max }}\left[ \frac{{\alpha {\lambda _{\min }}{{(P)}^{\frac{{1 - \alpha }}{{2\alpha  - 1}}}}}}{{2\alpha  - 1}}(R + {R^\top}) + {a^{ - 1}}{S^\top}S \right] \\&
+ \frac{{bd}}{2}\left( {{{[n(N - 1)]}^{\frac{{\alpha  - 1}}{\alpha }}} + 1} \right) + a{\left( {\frac{\alpha {\lambda _{\max }}{{(P)}^{\frac{{1 - \alpha }}{{2\alpha  - 1}}}}}{{2\alpha  - 1}}} \right)^2}\\&\times{\left[ {n(N - 1)} \right]^{\frac{{\alpha  - 1}}{\alpha }}}+1\} V{(t)^{\frac{1}{\alpha }}}
 + {\left( {{e^\top}(t)e(t)} \right)^{\frac{{2 - 2\alpha }}{{2\alpha  - 1}}}}\\&\times({e^\top}(t)(\gamma^{2}-1)^{-1} Qe(t))
 - (\sqrt{\gamma^2-1} w(t) - {e^\top}(t)\\&\times{\left( {{e^\top}(t)e(t)} \right)^{\frac{{1 - \alpha }}{{2\alpha  - 1}}}}\frac{(M \otimes {I_n})G(X)}{\sqrt{\gamma^2-1}})^\top (\sqrt{\gamma^2-1} w(t)- {e^\top}(t)\\&\times{\left( {{e^\top}(t)e(t)} \right)^{\frac{{1 - \alpha }}{{2\alpha  - 1}}}}\frac{(M \otimes {I_n})G(X)}{\sqrt{\gamma^2-1}})]dt
 + V(0) - V(T).
	\end{split}
	\end{equation}
therefore, we obtain	
	\begin{equation}
	\begin{split}
	J \le & \int_0^T  {qV(t)^{\frac{1}{\alpha}}dt} + V(0) - V(T).
	\end{split}
	\end{equation}
	Using the state zero response and if (\ref{ineqq2}) is satisfied, it is concluded that $J\le 0$. Therefore, the global ${H_\infty }$ finite-time consensus is obtained. Note that $q$ depends on the upper-bound of delays (i.e., $d$). Hence, the control criterion is delay-dependent.
\end{proof}
\begin{remark}
	As has been mentioned, the disturbance matrix $G(X)$ is assumed to satisfy the inequality (\ref{ineqq1}). This condition includes the constant disturbance matrices and the state-dependent bounded ones. A physical example for $G(X)$ will be introduced in simulations, which depends on the states and models the friction.  
\end{remark}
			\begin{remark}
				Whenever accessibility to all states of the neighbor agents is not possible, the control signal can be changed to the following:
				\begin{equation}\label{define}	
				\begin{split}
				&U(t) = - {\Phi ^{ - 1}}(X(t))[F(X(t))+ ({I_N} \otimes {K_1})\\&\times{({X^\top}(t)X(t))^{\frac{{1 - \alpha }}{{2\alpha  - 1}}}}X(t)+({I_N} \otimes {K_3})(L \otimes {I_l})({X^\top}(t - d)\\&\times({I_N} \otimes C^\top C)X(t - d))^{\frac{{1 - \alpha }}{{2\alpha  - 1}}}({I_N} \otimes C)X(t - d)],
				\end{split}
				\end{equation}
				
				where we employ a new control gain ${K_3} \in {\mathbb{R}^{n \times l}}$ with $l$ as the number of accessible neighbor outputs. In addition, $C \in {\mathbb{R}^{l \times n}}$ is the output matrix of agents. In this case, the matrix $S$ in (\ref{ineq}) should be substituted by $S_1 =  - (M \otimes {I_n})({I_N} \otimes {K_3})(L \otimes {I_l})({I_N} \otimes C){(M \otimes {I_n})^ + }$ and the $P$ matrix in the third and forth terms of (\ref{pred err dynN3}) is as $P_1={{{(M \otimes {I_n})}^ + }^\top({I_N} \otimes C^\top C){{(M \otimes {I_n})}^ + } }$ .
				With these changes, Theorem \ref{Theorem1} holds.
			\end{remark}
		\subsection{Leader-Follower Finite-Time Consensus}
			If the agents are being determined to make consensus on tracking a specific path, the problem can be considered in a leader-follower topology in which, the leader agent has the accessibility to the tracking target and other agents' goal is to track the leader in finite-time.
			Without loss of generality, consider the first agent as the leader. Therefore, the target tracking error for the leader can be considered as
			\begin{equation}\label{leader_err}
			{e_l}(t)=x_1(t)-r(t),
			\end{equation}
			in which, $r(t)$ is the tracking reference of the multi-agent system.
			In addition consider the control input signal of the leader as
		\begin{equation}\label{leader_cont}
		\begin{array}{l}
		u_1(t)=- {\Phi ^{ - 1}}(X(t))[F(X(t))\\[3pt]-K_3{\lambda _{\min }}{(P)^{\frac{{1 - \alpha }}{{2\alpha  - 1}}}}e_l(t){(e_l{(t)^\top}e_l(t))^{\frac{{1 - \alpha }}{{2\alpha  - 1}}}}+\dot r(t)],
		\end{array}
		\end{equation}
			where $K_3 \in \mathbb{R}^{n \times n}$ is the control gain for the leader. 
			Therefore, the consensus error vector for the leader-follower case can be defined as
			\begin{equation}\label{leader_err_vec}
			\xi (t) = {\left[ {\begin{array}{*{20}{c}}
					{{e_l}{{(t)}^\top}}&{e_f{{(t)}^\top}}
					\end{array}} \right]^\top}.
			\end{equation}
			In the above definition, $e_f(t)$ is defined the same as (\ref{consensus vec}).
			By this definition and considering the same control signal for the follower agents as (\ref{controller}), the dynamics of each parts of the new consensus error vector is achieved as
			\begin{align*}
			\begin{array}{l}
			{{\dot e}_l}(t) = -{K_3}{\lambda _{\min }}{(P)^{\frac{{1 - \alpha }}{{2\alpha  - 1}}}} e_l(t){\left(e_l{(t)^\top}e_l(t)\right)^{\frac{{1 - \alpha }}{{2\alpha  - 1}}}}\\+({I_n } \otimes{\left[ {\begin{array}{*{20}{c}}
					1&0& \cdots &0
					\end{array}} \right]})G(X)w(t),\\
			\dot e_f(t) = -(M \otimes {I_n})(({K_3}{\lambda _{\min }}(P)^{\frac{{1 - \alpha }}{{2\alpha  - 1}}})\\
			\otimes {\left[ \begin{array}{*{20}{c}}
					1&0& \cdots &0
					\end{array} \right]^\top})e_l(t) 
				\left(e_l(t)^\top e_l(t)\right)^{\frac{{1 - \alpha }}{{2\alpha  - 1}}}\\
				+ (M \otimes {I_n})
				(\dot r(t) \otimes {\left[ \begin{array}{*{20}{c}}
				1&0& \cdots &0
					\end{array} \right]^\top})
			- (M \otimes {I_n})\\\times({D_{\bar u}} \otimes {K_1}){{(M \otimes {I_n})}^ + } e_f(t)
			\left( {{e_f^\top}(t)Pe_f(t)} \right)^{\frac{{1 - \alpha }}{{2\alpha  - 1}}}\\
			- \left[{(M \otimes {I_n})({D_{\bar u}} \otimes {K_2})(L \otimes {I_n})}{{{(M \otimes {I_n})}^ + }} \right]e_f(t - d)\\\times\left({e_f^\top}(t - d) Pe_f(t - d)\right)^{\frac{{1 - \alpha }}{{2\alpha  - 1}}}+ (M \otimes {I_n})G(X)w(t),
			\end{array}
			\end{align*}
		in which, ${D_{\bar u}} = \diag(0,1,...,1)$.
			
			In addition, using (\ref{leader_err}) and (\ref{consensus vec}), $\dot r(t)$ can be computed as
			\begin{align*}
			\begin{array}{l}
			\dot r(t) = ({I_n } \otimes{\left[ {\begin{array}{*{20}{c}}
					1&0& \cdots &0
					\end{array}} \right]})(M \otimes {I_n})^ {+}\dot e_f(t)\\+{K_3} {\lambda _{\min }}{(P)^{\frac{{1 - \alpha }}{{2\alpha  - 1}}}}e_l(t) {\left(e_l{(t)^\top}e_l(t)\right)^{\frac{{1 - \alpha }}{{2\alpha  - 1}}}}\\-\left({I_n } \otimes{\left[ {\begin{array}{*{20}{c}}
					1&0& \cdots &0
					\end{array}} \right]}\right)G(X)w(t).
				\end{array}
			\end{align*}	
			After a few simplifications, $\dot \xi (t)$ is obtained as
			\begin{equation}\label{final_dot_e}
\begin{split}
&\dot \xi (t) =A \left[ {\begin{array}{*{20}{c}}
{\lambda _{\min }}{(P)^{\frac{{1 - \alpha }}{{2\alpha  - 1}}}} {{e_l}(t){{({e_l}{{(t)}^\top}{e_l}(t))}^{\frac{{1 - \alpha }}{{2\alpha  - 1}}}}}\\
{{e_f}(t){{\left( {e_f^\top(t)P{e_f}(t)} \right)}^{\frac{{1 - \alpha }}{{2\alpha  - 1}}}}}
\end{array}} \right]\\&
 + B \left[ {\begin{array}{*{20}{c}}
{{e_l}(t - d){{({e_l}{{(t - d)}^\top}{e_l}(t - d))}^{\frac{{1 - \alpha }}{{2\alpha  - 1}}}}}\\
{{e_f}(t - d){{\left( {e_f^\top(t - d)P{e_f}(t - d)} \right)}^{\frac{{1 - \alpha }}{{2\alpha  - 1}}}}}
\end{array}} \right]
\\& + T G(X)w(t).
\end{split}
			\end{equation}
In the above expression, $A = \left[ {\begin{array}{*{20}{c}}
-K_3&0\\
0&{{A_1}}
\end{array}} \right]$, $B = \left[ {\begin{array}{*{20}{c}}
0&0\\
0&{{B_1}}
\end{array}} \right]$, $T=\left[ {\begin{array}{*{20}{c}}
{{T_1}}\\
{{T_2}}
\end{array}} \right]$,
\newline
where 
\begin{align*}
\begin{split}
&{A_1} = {({I_{n(N - 1)}} - (M \otimes {I_n})({D_u} \otimes {I_n}){(M \otimes {I_n})^ + })^{ - 1}}\\&\times(M \otimes {I_n})[-({D_{\bar u}} \otimes {K_1}) + ({D_u} \otimes {I_n})]{(M \otimes {I_n})^ + },\\&
{B_1} = {({I_{n(N - 1)}} - (M \otimes {I_n})({D_u} \otimes {I_n}){(M \otimes {I_n})^ + })^{ - 1}}\\&\times\left[ -{(M \otimes {I_n})({D_{\bar u}} \otimes {K_2})(L \otimes {I_n}){{(M \otimes {I_n})}^ + }} \right],\\&
{T_1} = {I_n} \otimes \left[ {\begin{array}{*{20}{c}}
1&0& \cdots &0
\end{array}} \right],\\&
{T_2} = {({I_{n(N - 1)}} - (M \otimes {I_n})({D_u} \otimes {I_n}){(M \otimes {I_n})^ + })^{ - 1}}\\&\times(M \otimes {I_n})[{I_{nN}} - ({D_u} \otimes {I_n})],
\end{split}
\end{align*}
	and ${D_u} = \diag(1,0,...,0)$.
	
By using the same Lyapunov function as (\ref{lyap}) with the variable $\xi(t)$ instead of $e(t)$ in the leader-less case, we get
	\begin{align*}
	\begin{split}
	&\dot V(t) \le  {\xi^\top}(t){({\xi^\top}(t)\xi(t))^{\frac{{1 - \alpha }}{{2\alpha  - 1}}}}[\frac{\alpha }{{2\alpha  - 1}}{\lambda _{\min }}{(P)^{\frac{{1 - \alpha }}{{2\alpha  - 1}}}}(A + {A^\top})\\& + a^{-1}{B^\top}B]\xi(t){({\xi^\top}(t)\xi(t))^{\frac{{1 - \alpha }}{{2\alpha  - 1}}}}
	+ a{(\frac{\alpha }{{2\alpha  - 1}}{\lambda _{\max }}{(P)^{\frac{{1 - \alpha }}{{2\alpha  - 1}}}})^2}\\&\times{\xi^\top}(t - d ){({\xi^\top}(t - d )\xi(t - d ))^{\frac{{1 - \alpha }}{{2\alpha  - 1}}}}\xi(t - d ){({\xi^\top}(t - d )\xi(t - d ))^{\frac{{1 - \alpha }}{{2\alpha  - 1}}}}
	\\&+ {\left( {{\xi^\top}(t)\xi(t)} \right)^{\frac{{1 - \alpha }}{{2\alpha  - 1}}}}({\xi^\top}(t)TG(X)w(t)
	+ {w(t)^\top}G{(X)^\top}{T^\top}\xi(t)).
	\end{split}
	\end{align*}
			From now on, the previous section proof can be followed.
		\section{Finite-Time Stochastic Consensus}\label{stoch}
		In this section, a multi-agent system subject to stochastic disturbances is considered. Due to the random uncertainties such as stochastic forces on physical systems, noisy measurements or stochastic couplings in the interconnected systems, this approach provides more  applicability in real world. A complex network exhibits complicated dynamics which may be absolutely different from those of a single
		node.  Since the whole multi-agent system becomes a coupled system influenced by stochastic perturbations, a stochastic differential equation should be considered for that, instead of a deterministic one. By finite-time stochastic consensus, we mean the expected value of the consensus error norm converges to zero in finite-time. Consider the stochastically perturbed model of the nonlinear multi-agent system as
		\begin{equation}\label{stoch_model}
		\begin{split}
dX(t) =&(F(X(t)) + \Phi (X(t))U(t))dt+ H(X)d\rho(t),
\end{split}
		\end{equation}
		where $\rho(t)=[\rho_1(t),...,\rho_N(t)]$ is a vector Wiener process and $H(X)$ is a bounded matrix. This model considers a generic coupling with the disturbance and provides a sufficiently general and real representation for the system, which can also include the interconnected network systems subject to stochastic couplings. In addition, the Wiener process satisfies the following properties.
		\begin{equation}\label{wienr_prop}
\begin{array}{l}
E\left[ {H(X)d\rho} \right] = 0,\\
E\left[ {{{(H(X)d\rho)}^\top}(H(X)d\rho)} \right] = \trace \left[ {{H(X)^\top}H(X)} \right]dt.
\end{array}
		\end{equation}
	
		\begin{assumption}\label{ass6}
		   	It is assumed that $H(X)$ for	$X(t) \in \Omega  \subseteq \mathbb{R}^{nN}$ satisfies
		   	\begin{equation*}
		   	    	\trace \left[ {{H(X)^\top}H(X)} \right] \le \left\| X \right\|_2^{\frac{\alpha }{{2\alpha  - 1}}}.
		   	\end{equation*}
		\end{assumption}
		\begin{definition}
		    The stochastic multi-agent system (\ref{stoch_model}) is said reach consensus stochastically in finite-time if, for a  suitable control signal, there exists a constant $T_{c1}>0$ such that $\mathop {\lim E{{\left\| {e(t)} \right\|}_2}}\limits_{t \to {T_{c1}}}  = 0$ and $E{\left\| {e(t)} \right\|_2} = 0$ for $t>T_{c1}$.
		\end{definition}
		\begin{theorem}
	Consider the nonlinear stochastically perturbed multi-agent system (\ref{stoch_model}) with the consensus error vector (\ref{consensus vec}), satisfying Assumptions \ref{ass1}, \ref{asss}, \ref{ass3}, \ref{ass6}. Stochastic finite-time delay-dependent consensus is obtained utilizing the  control algorithm (\ref{controller}), provided that there exist positive scalars $a, b$ such that the following conditions are satisfied.
			\begin{equation}\label{a}
			\begin{split}
&q_1={\lambda _{\max }}\left[ \frac{\alpha }{{2\alpha  - 1}}{\lambda _{\min }}{{(P)}^{\frac{{1 - \alpha }}{{2\alpha  - 1}}}}(R + {R^\top})+ {a^{ - 1}}{S^\top}S \right]\\& +\frac{{{\lambda _{\max }}({M^\top}M)}}{{{{\left[ {{\lambda _{\min }}({M^{+\top}}{M^ + })} \right]}^{\frac{\alpha }{{2\alpha  - 1}}}}}}+ \frac{{bd}}{2}\left( {{{[n(N - 1)]}^{\frac{{\alpha  - 1}}{\alpha }}} + 1} \right)
\\& + a{\left( {\frac{\alpha }{{2\alpha  - 1}}{\lambda _{\max }}{{(P)}^{\frac{{1 - \alpha }}{{2\alpha  - 1}}}}} \right)^2}{\left[ {n(N - 1)} \right]^{\frac{{\alpha  - 1}}{\alpha }}} < 0,
			\end{split}
			\end{equation}
			\begin{equation}\label{b}
				{{\lambda _{\min }}(P)(R + {R^\top}){ - b{I_{n(N - 1)}}}} \le 0.
			\end{equation}
		Moreover, the finite consensus time $T_{c1}$ is bounded by
	\begin{equation}\label{settle_2}
			    	{{{T}}_{c1}} \le \frac{{ \alpha }}{{\left|q_1\right|(\alpha  - 1)}}V{(0)^{\frac{{\alpha  - 1}}{\alpha }}}.
			\end{equation}
		\end{theorem}
\begin{proof}
        Consider the system dynamics (\ref{stoch_model}) and consensus error definition of (\ref{consensus vec}). Consider the Lyapunov function mentioned in (\ref{lyap}).
        Taking into account the properties of the Wiener process and differentiating the Lyapunov function, we obtain
        \begin{equation}\label{der_lyap}
        \begin{split}
        &dV(t) = \frac{\alpha }{{2\alpha  - 1}}{\left( {{e^\top}(t)e(t)} \right)^{\frac{{1 - \alpha }}{{2\alpha  - 1}}}}\left( {e^\top}(t)de(t)\right.\\&\left. + d{e^\top}(t)e(t) + d{e^\top}(t)de(t) \right).
        \end{split}
        \end{equation}
        Substituting the control signal and getting the expectations of both sides, gives 
        \begin{equation}\label{der_lyap2}
\begin{split}
&E\left[ {dV(t)} \right]= E[\frac{\alpha }{{2\alpha  - 1}}({e^\top}(t){\left( {{e^\top}(t)e(t)} \right)^{\frac{{1 - \alpha }}{{2\alpha  - 1}}}}Re(t)\\&\times{\left( {{e^\top}(t)Pe(t)} \right)^{\frac{{1 - \alpha }}{{2\alpha  - 1}}}}+ {e^\top}(t)\left( {e^\top}(t)Pe(t) \right)^{\frac{{1 - \alpha }}{{2\alpha  - 1}}}
 {R^\top}e(t)\\&\times {\left( {{e^\top}(t)e(t)} \right)^{\frac{{1 - \alpha }}{{2\alpha  - 1}}}}
 + {e^\top}(t){\left( {{e^\top}(t)e(t)} \right)^{\frac{{1 - \alpha }}{{2\alpha  - 1}}}}Se(t - d)\\&\times{({e^\top}(t - d)Pe(t - d))^{\frac{{1 - \alpha }}{{2\alpha  - 1}}}}
 + {e^\top}(t - d)\\&\times{\left( {{e^\top}(t - d)Pe(t - d)} \right)^{\frac{{1 - \alpha }}{{2\alpha  - 1}}}}{S^\top}e(t){\left( {{e^\top}(t)e(t)} \right)^{\frac{{1 - \alpha }}{{2\alpha  - 1}}}}
\\& + {\left( {{e^\top}(t)e(t)} \right)^{\frac{{1 - \alpha }}{{2\alpha  - 1}}}} \trace\{ {H^\top}{(M \otimes {I_n})^\top}(M \otimes {I_n})H\} )]dt.
\end{split}
        \end{equation}
        Considering $H(X)$ satisfying Assumption \ref{ass6} and utilizing the Wiener process properties and the results of previous sections, (\ref{der_lyap2}) can be written as
        \begin{equation}\label{der_lyap3}
\begin{split}
&E\left[ {\dot V(t)} \right]\le \{ {\lambda _{\max }}\left[ \frac{\alpha }{{2\alpha  - 1}}{\lambda _{\min }}{{(P)}^{\frac{{1 - \alpha }}{{2\alpha  - 1}}}}(R + {R^\top})\right.\\&\left.+ {a^{ - 1}}{S^\top}S \right]
 + \frac{{{\lambda _{\max }}({M^\top}M)}}{{{{\left[ {{\lambda _{\min }}({M^{ +\top}}{M^ + })} \right]}^{\frac{\alpha }{{2\alpha  - 1}}}}}}
 + \frac{{bd}}{2}\\&\times\left( {{{[n(N - 1)]}^{\frac{{\alpha  - 1}}{\alpha }}} + 1} \right)
 + a{\left( {\frac{\alpha }{{2\alpha  - 1}}{\lambda _{\max }}{{(P)}^{\frac{{1 - \alpha }}{{2\alpha  - 1}}}}} \right)^2}\\&\times{\left[ {n(N - 1)} \right]^{\frac{{\alpha  - 1}}{\alpha }}}\} E{{\rm{[}}V(t){\rm{]}}^{\frac{1}{\alpha }}}
 +E [\int_{t - d}^t {\zeta {{(s)}^\top}{\rm{D}}\zeta (s)ds}],
\end{split}
			\end{equation}
			where
			\begin{align*}
			\begin{array}{l}
			\zeta (s) = {\left[ {\begin{array}{*{20}{c}}
					{\Xi (s)}&{\Xi (s - d)}
					\end{array}} \right]^\top},\\
			D = \left[ {\begin{array}{*{20}{c}}
				{{\lambda _{\min }}(P)(R + {R^\top}){ - b{I_{n(N - 1)}}}}&{{\lambda _{\max }}(P)S}\\
				{{\lambda _{\max }}(P){S^\top}}&0
				\end{array}} \right],
			\end{array}
			\end{align*}
			and	$\Xi (s) = {e^\top}(s){({e^\top}(s)e(s))^{\frac{{1 - \alpha }}{{2\alpha  - 1}}}}$.
			           
     Therefore, provided that the sufficient conditions of (\ref{a}) and (\ref{b}) are satisfied, finite-time convergence of $E\left[V(t)\right]$ to zero and consequently  $\mathop {\lim E{{\left\| {e(t)} \right\|}_2}}\limits_{t \to {T_{c1}}}  = 0$ is guaranteed. Again in (\ref{a}), a dependency to the parameter $d$ is apparent which makes the algorithm delay dependent.
\end{proof}
	
		\section{Simulation Results}\label{sim}
	In this section, to evaluate the effectiveness of the presented approach, simulations are performed on a group of four identical nonholonomic mobile robots as the agents (Figure \ref{figurelabel}). 
	The mobile robots consist of a vehicle with two driving wheels  on the same axis and a front free wheel. The motion and orientation are achieved by independent actuators, e.g., dc motors which provide the necessary torques to the rear wheels. 
	The nonlinear dynamics of each robot in an $n$ dimensional configuration space with coordinates $(q_1,...,q_n)$ and subject to $c$ constraints are described as 
	\begin{equation}\label{dyn}
	\begin{array}{l}
	M_c(q)\ddot q + V_c(q,\dot q)\dot q + F_c(\dot q) + {\tau _d}
	= B_c(q)\tau_c  - {A_c^\top}(q)\mu,
	\end{array}
	\end{equation}
	with
	\begin{equation}\label{dyn2}
	\begin{array}{l}
	M_c(q) = \left[ {\begin{array}{*{20}{c}}
		m&0&{mp\sin \theta }\\
		0&m&{ - mp\cos \theta }\\
		{mp\sin \theta }&{ - mp\cos \theta }&1
		\end{array}} \right],\\
	V_c(q,\dot q) = \left[ {\begin{array}{*{20}{c}}
		0&0&{mp\dot \theta \cos \theta }\\
		0&0&{mp\dot \theta \sin \theta }\\
		0&0&0
		\end{array}} \right], \tau_c  = \left[ {\begin{array}{*{20}{c}}
		{{\tau _1}}\\
		{{\tau _2}}
		\end{array}} \right],\\ B_c(q) = \frac{1}{r}\left[ {\begin{array}{*{20}{c}}
		{\cos \theta }&{\cos \theta }\\
		{\sin \theta }&{\sin \theta }\\
		R&{ - R}
		\end{array}} \right],
	A_c^\top(q) = \left[ {\begin{array}{*{20}{c}}
		{ - \sin \theta }\\
		{\cos \theta }\\
		{ - p}
		\end{array}} \right],\\
 \mu =  - m({{\dot x}_c}\cos \theta  + {{\dot y}_c}\sin \theta )\dot \theta,
	\end{array}
	\end{equation}
	where $M_c(q) \in {\mathbb{R}^{n \times n}}$ is a symmetric positive definite inertia matrix, $V_c(q,\dot q) \in {\mathbb{R}^{n \times n}}$ is the coriolis and centripetal matrix, $F_c(\dot q) \in {\mathbb{R}^{n}}$ stands for the surface friction,  ${\tau _d}$ is the bounded unknown disturbance, $B_c(q) \in {\mathbb{R}^{n}}$ indicates the input transformation matrix, $\tau_c  \in {\mathbb{R}^{n}}$ is the input vector and $A_c(q) \in {\mathbb{R}^{c \times n}}$, $\mu  \in {\mathbb{R}^{c}}$ are the constraint associated matrix and vector force, respectively.
	
	Furthermore, kinematic equality constraints can be expressed as
	\begin{equation}\label{1st}
	A_c(q)\dot q = 0.
	\end{equation}
	There is a full rank matrix $S(q) \in {\mathbb{R}^{n \times (n-c)}}$ constituting of a set of smooth linearly independent vector fields that spans the null space of $A_c(q)$, i.e.,
	\begin{equation}\label{2nd}
	S^\top(q)A_c^\top(q)=0.
	\end{equation}
	According to (\ref{1st}) and (\ref{2nd}), a vector time function $\nu (t) \in {\mathbb{R}^{ n-c}}$ could be found such that for all $t$
	\begin{equation*}
	\dot q=S(q)\nu (t).
	\end{equation*}
	The nonholonomic constraint states that the robot can only move in the direction normal to the axis of the driving wheels.  If $r=n-c$, after transformation from $q$ coordinates to $\nu$ configuration, the new input matrix $\bar B_c$ is constant and nonsingular. In this case, the mentioned matrices and transformed dynamics can be described as
	\begin{align*}
	\begin{array}{l}
	\nu  = \left[ {\begin{array}{*{20}{c}}
		{{\nu_1}}\\
		{{\nu_2}}
		\end{array}} \right] = \left[ {\begin{array}{*{20}{c}}
		v\\
		\omega 
		\end{array}} \right],\\
	\dot q = \left[ {\begin{array}{*{20}{c}}
		{{{\dot x}_c}}\\
		{{{\dot y}_c}}\\
		{\dot \theta }
		\end{array}} \right] = \left[ {\begin{array}{*{20}{c}}
		{\cos \theta }&{ - p\sin \theta }\\
		{\sin \theta }&{p\cos \theta }\\
		0&1
		\end{array}} \right]\left[ {\begin{array}{*{20}{c}}
		v\\
		\omega 
		\end{array}} \right],\\
	\bar B_c = {S^\top}B_c = \frac{1}{r}\left[ {\begin{array}{*{20}{c}}
		1&1\\
		R&{ - R}
		\end{array}} \right],
	\end{array}
	\end{align*}
	\begin{align*}
		({S^\top}{M_c}S)\dot \nu + {S^\top}({M_c}\dot S + {V_c}S)\nu + \bar F(\nu) + {{\bar \tau }_d} = {S^\top}{B_c}{\tau _c},
	\end{align*}
	where $v$ and $w$ are the linear and angular velocities of the mobile robot. Besides, $x_c$ and $y_c$ are indications of its position and $\theta$ the orientation of that. 
	The agents communicate with each other via the specified directed graph in Figure \ref{jo}. 
	The Laplacian matrix of the communication topology is
	\begin{equation*}
	{L} = \left[ {\begin{array}{*{20}{c}}
1&0&0&{ - 1}\\
{ - 1}&1&0&0\\
0&{ - 1}&1&0\\
0&0&{ - 1}&1
\end{array}} \right].
	\end{equation*}
	\begin{figure}
		\centering
		\includegraphics[width=4cm]{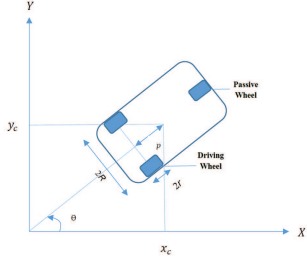}
		\caption{A nonholonomic mobile platform}
		\label{figurelabel}
	\end{figure}
			\begin{figure}
		\centering
		\includegraphics[width=3cm]{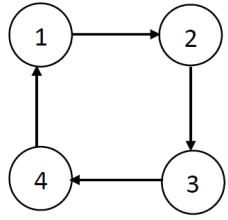}
		\caption{directed communication graph}
		\label{jo}
	\end{figure}
In this regard, the eigenvalues of the Laplacian matrix and the resulting matrix $M$ which has been used to construct the consensus error vector are as
\begin{equation*}
\begin{array}{l}
{\lambda } = {0,1+i, 1-i, 2};\\
l=[0.25, 0.25, 0.25, 0.25]^\top;\\
M = \left[ {\begin{array}{*{20}{c}}
	{ 0.0771}&{0.2992}&{ - 0.3927}&{0.0164}\\
	{0.1178}&{0.1751}&{ 0.1386}&{-0.4316}\\
	{ -0.4095}&{0.2594}&{ 0.1185}&{0.0316}
	\end{array}} \right].
\end{array}
\end{equation*}
	Therefore, it is concluded that $P=4 \times {I_6}$. The nominal parameters of the model are taken as ${m_{i}} = 10$ kg as the mass of vehicle, $R_i= 0.
	5$ m for the distance between driving wheels, $r_i = 0.05$ m for the radius of wheels and $p_i = 0.04$ m as the distance between the center of mass of the vehicle and wheels $(i=1,...,4)$. The initial conditions of agents are set as ${\nu}(1) = {[0.5,0.5]^\top}, {\nu}(2) = {[0.3,0.2]^\top},{\nu}(3) = {[0.8,0.1]^\top},{\nu}(4) = {[0.1,0.7]^\top
	}$. In addition, the delay between all the connected agents has been chosen as $(0.1 + 0.25{e^{ - t}})$s.
	\newline
	To obtain control signal parameters, the consensus conditions in (\ref{ineqq1})-(\ref{ineqq3}) were solved as LMI using CVX toolbox. We choose $\alpha=1.2$, $a=10$, $b=0.1$, $\mu=0.9$, $r=2$ and $G(X)=\diag(g_i(x_i))$, $i=1,...,4$, where $g_i(x_i)=(\frac{{x_i}_1^2}{1+{x_i}_1^2};\frac{{x_i}_2^2}{1+{x_i}_2^2})$, which is a function of velocities modeling the friction and $w(t)=1$. We also choose ${K_2} = \left[ {\begin{array}{*{20}{c}}
		{ - 5.14}&{  5.12}\\
		{3.55}&{ - 2.71}
		\end{array}} \right]$, since it is not a linear parameter in the inequalities. We obtain $\gamma_{min}=1.5$, ${K_1}=\left[ {\begin{array}{*{20}{c}}
{ - 28.36}&{  25.6}\\
{25.32}&{12.47}
\end{array}} \right]$, ${K_3}=\left[ {\begin{array}{*{20}{c}}
{ - 32.15}&{  26.21}\\
{23.54}&{14.77}
\end{array}} \right]$, $T_c \le 8.05$ s for the $H_{\infty}$ consensus. 
	\begin{figure*}
		\centering
		\begin{minipage}[b]{0.3\textwidth}
		\centering
			\includegraphics[width=\textwidth,height=4cm]{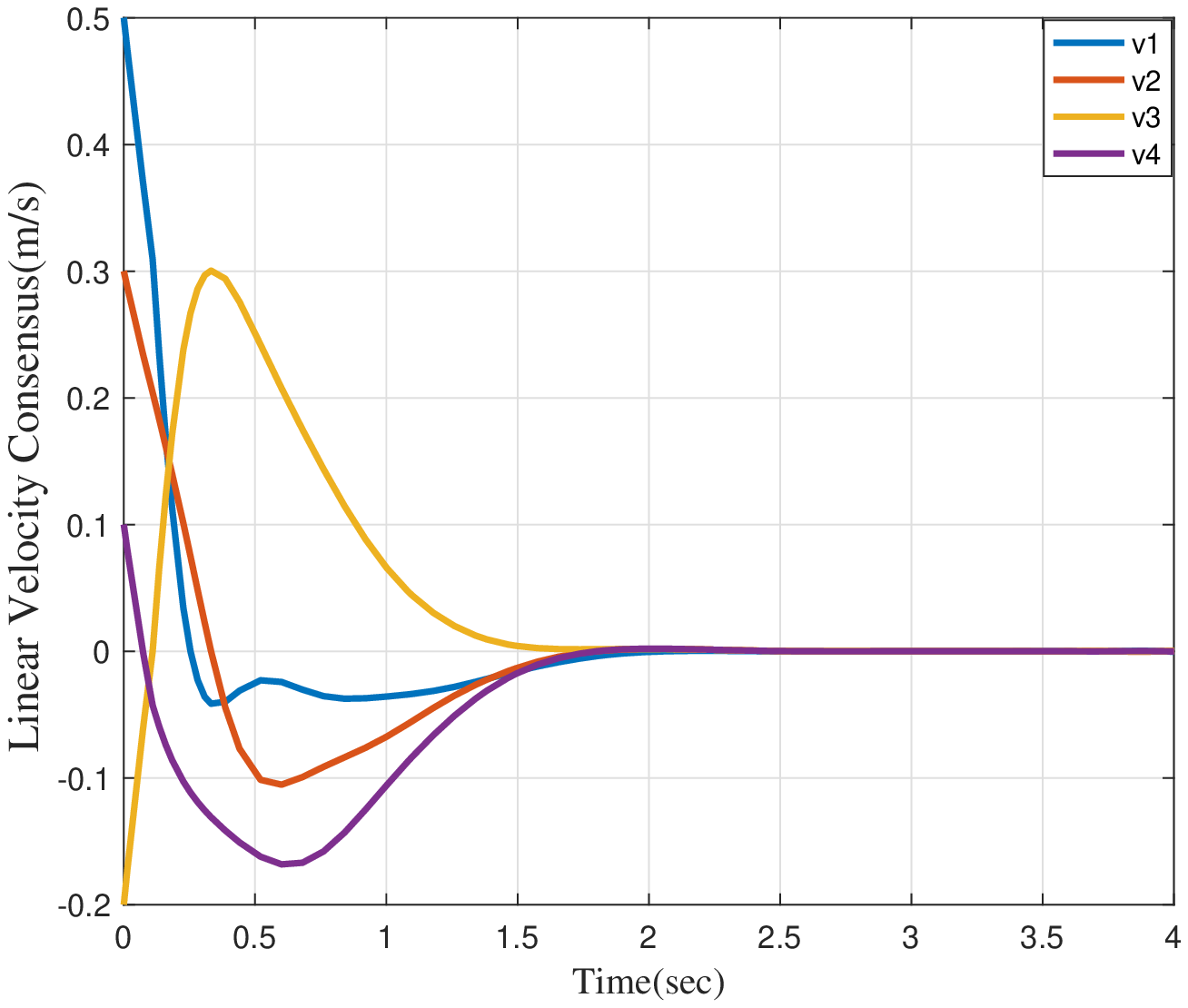}
			\subcaption(a) {Consensus of linear velocities}
			\label{v}
		\end{minipage}%
		\begin{minipage}[b]{0.3\textwidth}
		\centering
			\includegraphics[width=\textwidth,height=4cm]{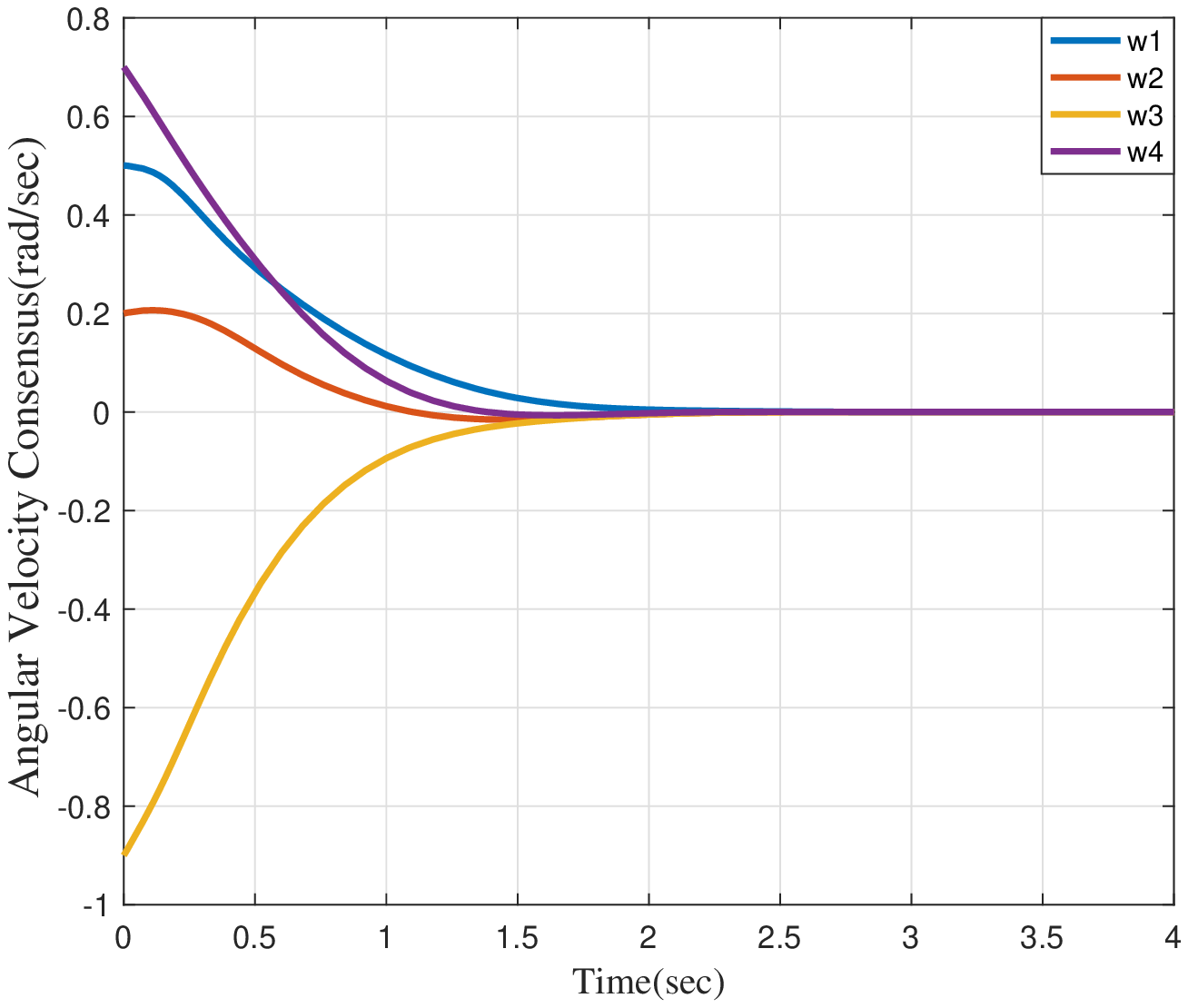}
			\subcaption(b) {Consensus of angular velocities}
			\label{w}
		\end{minipage}     
		\begin{minipage}[b]{0.3\textwidth}
		\centering
			\includegraphics[width=\textwidth,height=4cm]{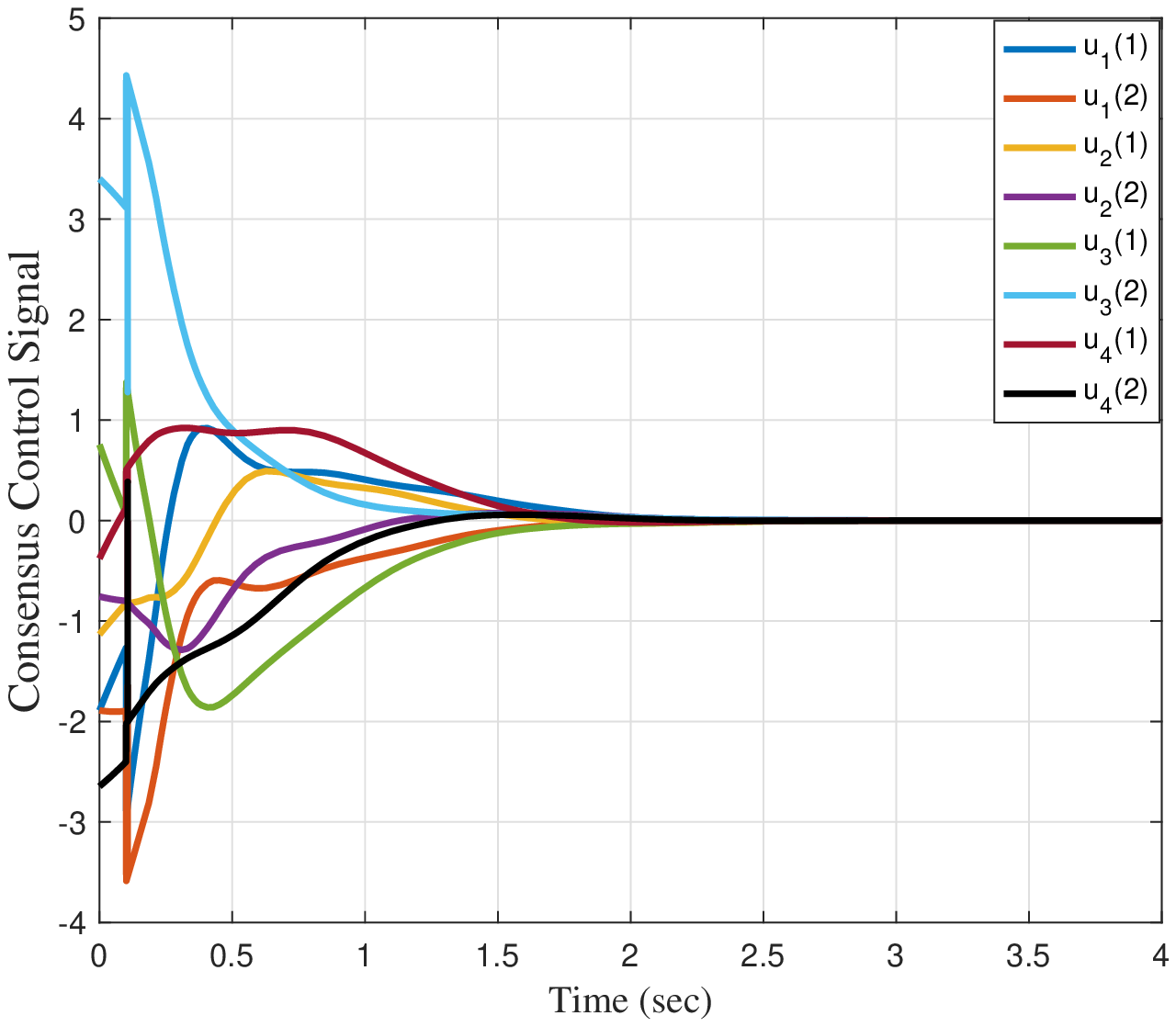}
			\subcaption(c) {Consensus control signals}
			\label{u}
		\end{minipage}
		\begin{minipage}[b]{0.3\textwidth}
		\centering
			\includegraphics[width=\textwidth,height=4cm]{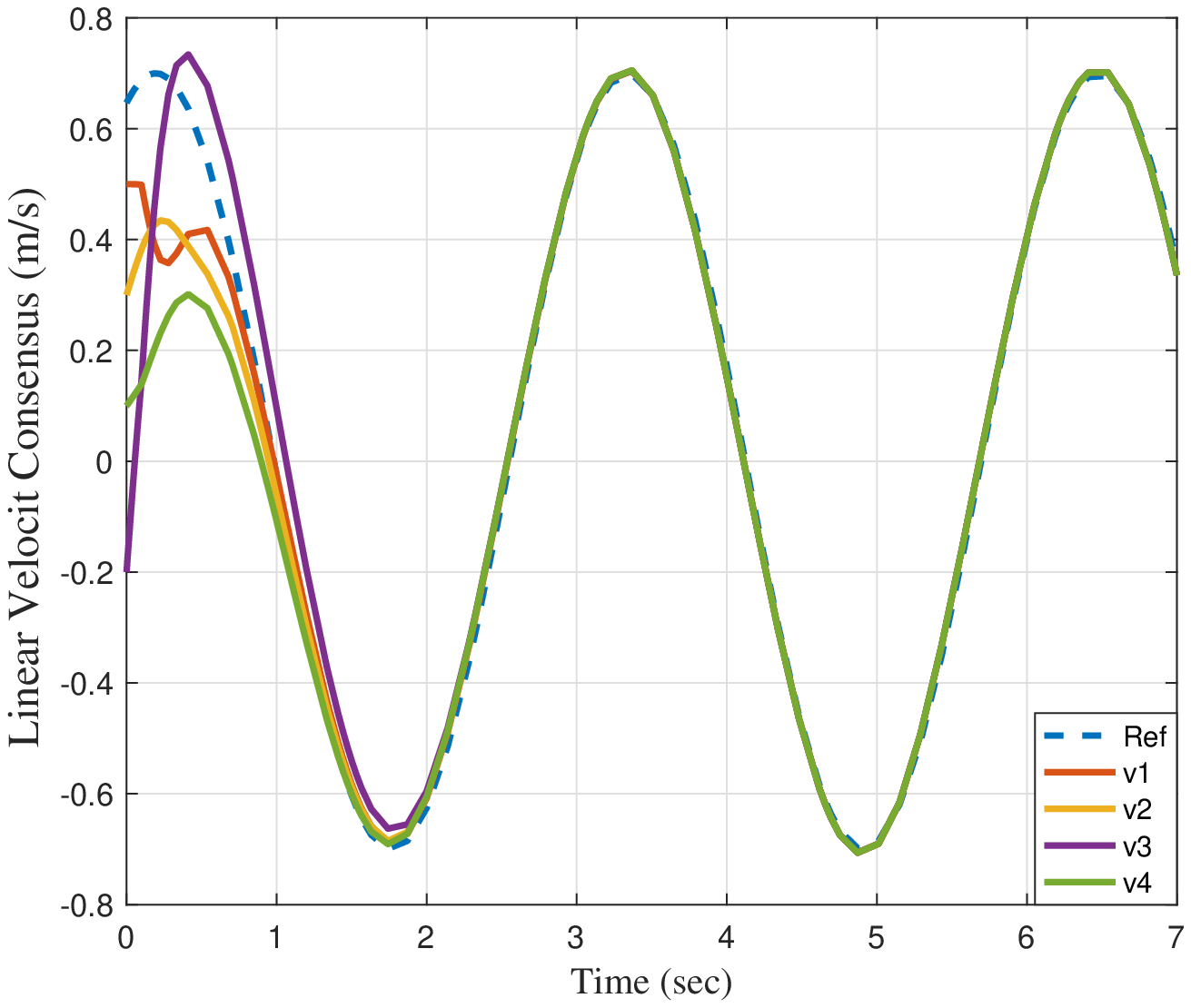}
			\subcaption(d) {Non-stationary target consensus of linear velocities}
			\label{v_station}
		\end{minipage}  
		\begin{minipage}[b]{0.3\textwidth}
		\centering
			\includegraphics[width=\textwidth,height=4cm]{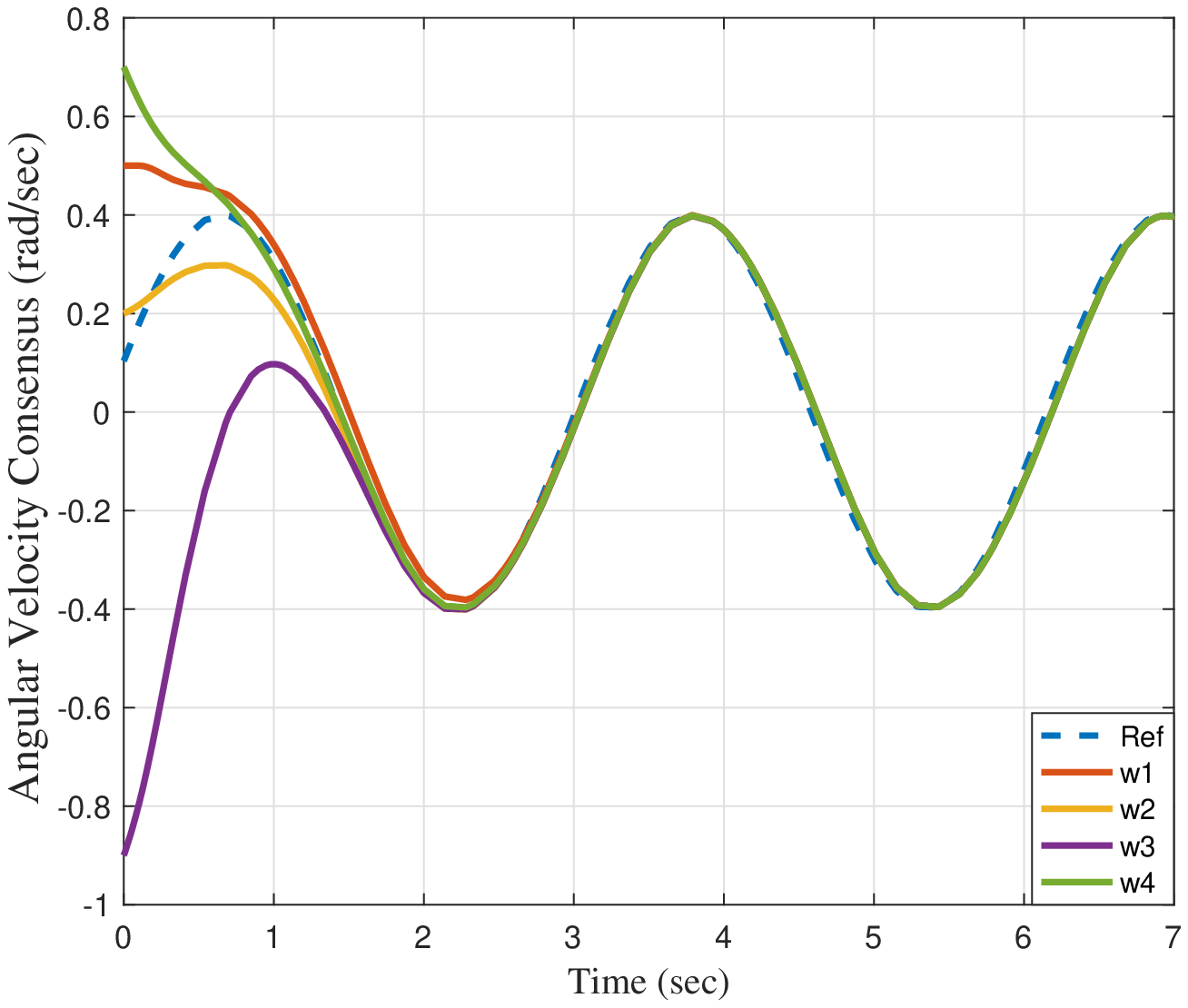}
			\subcaption(e) {Non-stationary target consensus of angular velocities}
			\label{w_station}
		\end{minipage}  
		\begin{minipage}[b]{0.3\textwidth}
		\centering
		\includegraphics[width=\textwidth,height=4cm]{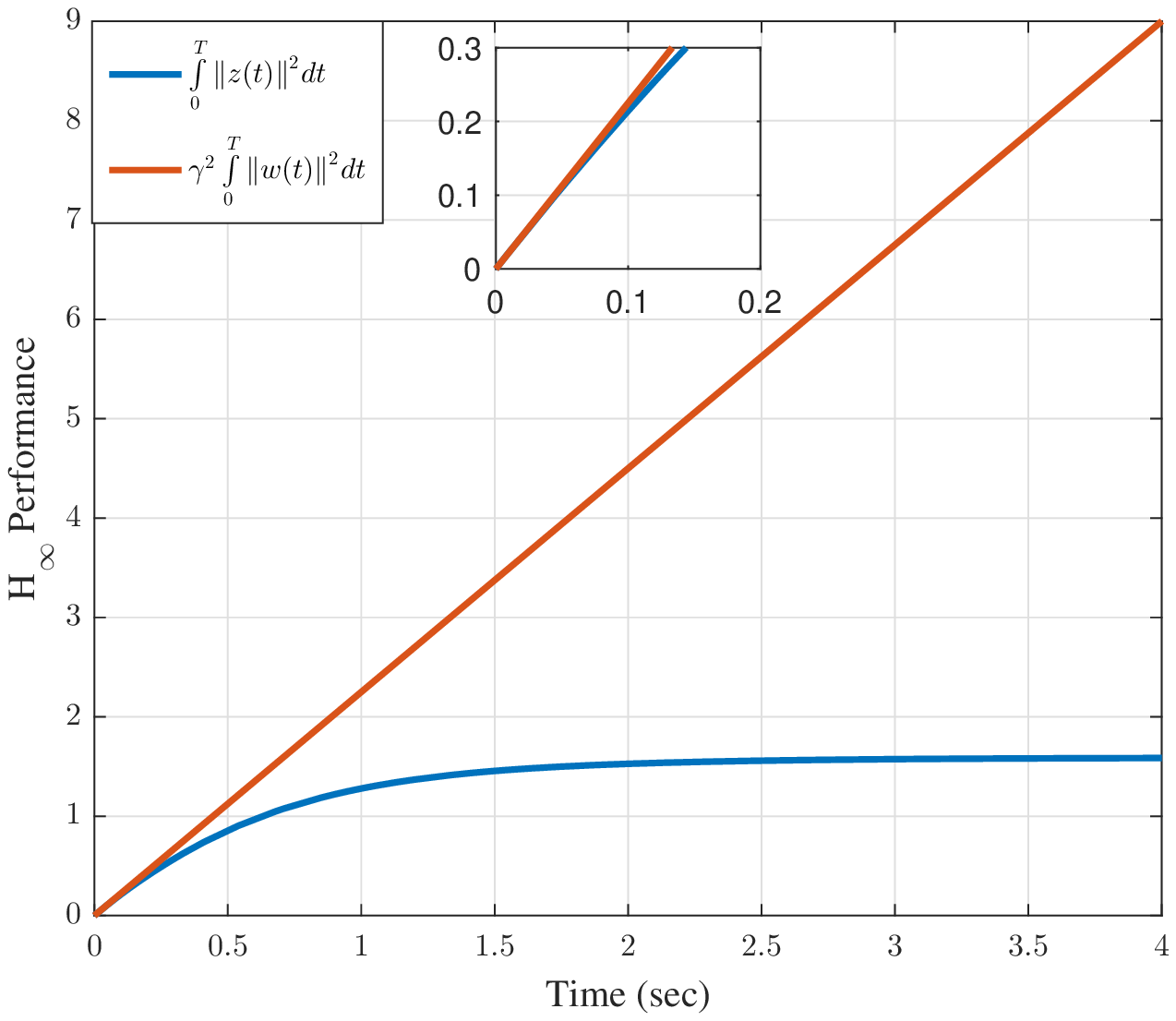}
		\subcaption(f) {$H_\infty$ performance}
		\label{per}
	\end{minipage}  
	\centering
		\caption{$H_\infty$ Finite-time consensus}\label{full}
	\end{figure*}   
	\begin{figure*}
		\centering
		\begin{minipage}[b]{0.3\textwidth}
		\centering
			\includegraphics[width=\textwidth,height=4cm]{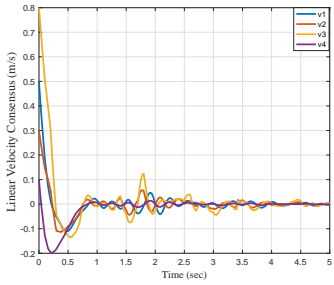}
			\subcaption(a) {Consensus of linear velocities}
			\label{v_part}
		\end{minipage}%
		\begin{minipage}[b]{0.3\textwidth}
		\centering
			\includegraphics[width=\textwidth,height=4cm]{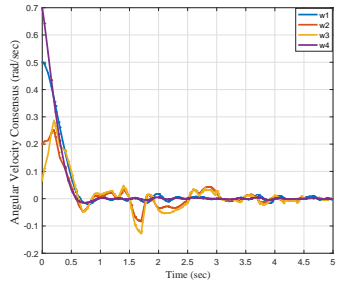}
			\subcaption(b) {Consensus of Angular velocities}
			\label{w_part}
		\end{minipage}
\begin{minipage}[b]{0.3\textwidth}
\centering
			\includegraphics[width=\textwidth,height=4cm]{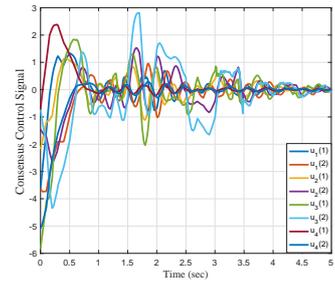}
			\subcaption(c) {Consensus control signals}	\label{u_part}
		\end{minipage}  
		\begin{minipage}[b]{0.3\textwidth}
		\centering
		\includegraphics[width=\textwidth,height=4cm]{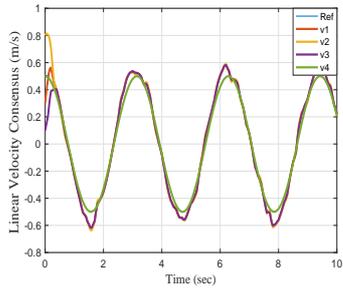}
		\subcaption(d) {Non-stationary target consensus of linear velocities}\label{v_leader}
	\end{minipage}  
	\begin{minipage}[b]{0.3\textwidth}
	\centering
	\includegraphics[width=\textwidth,height=4cm]{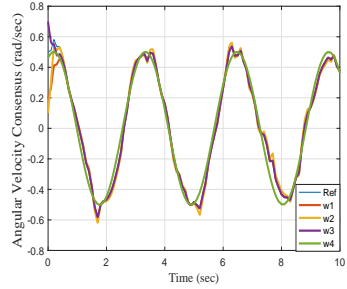}
	\subcaption(e) {Non-stationary target consensus of angular velocities}\label{w_leader}
\end{minipage}       
		\caption{Stochastic Finite-time consensus}\label{partial}
	\end{figure*} 
Figure \ref{full} shows the results for two cases of leaderless and non-stationary targets in which the fourth agent has been chosen as the reference of consensus scheme and $0.7\cos(2t-\pi/8)$, $0.4\sin(2t+\pi/12)$ are the targets for linear and angular velocities, respectively. 
For the stochastic consensus, $H(X)=\diag(h_i(x_i)), i=1,...,4$ with $h_i(x_i)=({x_i}_1^{\frac{\alpha}{2\alpha-1}};{x_i}_2^{\frac{\alpha}{2\alpha-1}})$, ${K_2} = \left[ {\begin{array}{*{20}{c}}
{ - 2.13}&{  3.86}\\
{3.52}&{ - 1.15}
\end{array}} \right]$  and a white noise with power of $0.1$ are chosen and control parameters are obtained as ${K_1} = \left[ {\begin{array}{*{20}{c}}
{ - 32.15}&{  29.22}\\
{20.31}&{17.94}
\end{array}} \right]$, ${K_3} = \left[ {\begin{array}{*{20}{c}}
{ - 34.54}&{  24.16}\\
{21.41}&{27.12}
\end{array}} \right]$. In addition, $T_{c1} \le 9.2$ s.
Figure \ref{partial} presents the stochastic consensus results. The reference $0.5\cos(2t)$ has been considered for the tracking scenario. It is worth noting that the leader-follower case formulation for stochastic consensus is similar to Section 4, in which $H(X)d\rho$ in (\ref{stoch_model}) is substituted for $G(X)w$ in (\ref{final_dot_e}). 
			\begin{remark}
			In comparison to other works, this paper gives a synchronous consensus which is ${\lim x_i(t)-x_j(t)= 0}$ and was not be achieved before in the presence of communication delays in a finite-time interval (see e.g.,  \cite{li2018finite}, \cite{sharifi2020finite}). This result is more promising, especially when the delays are long. Besides, providing the specific robust algorithms for deterministic and stochastic disturbances allows to confront them more effectively. 
	\end{remark}	
	
	\section{CONCLUSIONS}\label{conc}
	In this paper, the problem of robust finite-time consensus of a class of nonlinear systems in the presence of communication delays is considered. In order to compensate for the adverse effects of time-delays and uncertainties in a finite time interval, novel delay-dependent consensus algorithms have been suggested. In this regard, deterministic and stochastic uncertainties are considered and appropriate strategies to come up with them are proposed. The  algorithms can be used for cases with partial access of agents to their neighbor signals. In addition, there are not strong limitations on communication time-delays.  Future works include the robust finite-time consensus strategy for time-varying graph topologies.

					\bibliographystyle{IEEEtran}
					\bibliography{refss}   
					
				\end{document}